\documentclass[acmsmall,dvipsnames,svgnames,nonacm]{acmart}

\usepackage{multicol}
\usepackage{mathpartir}
\usepackage{tikz}
\usepackage{tikz-cd}
\usetikzlibrary{positioning,arrows.meta,fit,backgrounds,calc}
\usepackage[table]{xcolor}
\usepackage{pagecolor}
\usepackage{amsmath}
\usepackage{graphicx}
\usepackage{subcaption}
\usepackage{bookmark}
\usepackage[T1]{fontenc}
\usepackage[utf8]{inputenc}
\usepackage[scaled=0.82]{beramono}
\usepackage[yyyymmdd]{datetime}
\usepackage{minted}
\usepackage{fancyvrb}
\usepackage{thmtools}
\usepackage{nameref,cleveref}
\usepackage{listings}
\usepackage[inline]{enumitem}
\usepackage{everypage}
\definecolor{dkgreen}{rgb}{0.5,0.7,0.4}
\definecolor{ltblue}{rgb}{0,0.4,0.4}
\definecolor{dkblue}{rgb}{0,0.8,0.8}
\crefname{section}{\S}{\S\S}
\definecolor{dkviolet}{rgb}{0.3,0,0.5}
\definecolor{dkred}{rgb}{1,0,0}
\usepackage[yyyymmdd]{datetime}

\usepackage{hyperref}
\lstdefinelanguage{Coq}{
    mathescape=true,
    texcl=false,
    morekeywords=[1]{Section, Module, End, Require, Import, Export,
        Variable, Variables, Parameter, Parameters, Axiom, Hypothesis,
        Hypotheses, Notation, Local, Tactic, Reserved, Scope, Open, Close,
        Bind, Delimit, Definition, Let, Ltac, Fixpoint, CoFixpoint, 
        Morphism, Relation, Implicit, Arguments, Unset, Contextual,
        Strict, Prenex, Implicits, Inductive, CoInductive, Record,
        Structure, Canonical, Coercion, Context, Class, Global, Instance,
        Program, Infix, Theorem, Lemma, Corollary, Proposition, Fact,
        Remark, Example, Proof, Goal, Save, Qed, Defined, Hint, Resolve,
        Rewrite, View, Search, Show, Print, Printing, All, Eval, Check,
        Projections, inside, outside, Def},
    morekeywords=[2]{forall, exists, exists2, fun, fix, cofix, struct,
        match, with, end, as, in, return, let, if, is, then, else, for, of,
        nosimpl, when},
    morekeywords=[3]{Type, Prop, Set, true, false, option},
    morekeywords=[4]{pose, set, move, case, elim, apply, clear, hnf,
        intro, intros, generalize, rename, pattern, after, destruct,
        induction, using, refine, inversion, injection, rewrite, congr,
        unlock, compute, ring, field, fourier, replace, fold, unfold,
        change, cutrewrite, simpl, have, suff, wlog, suffices, without,
        loss, nat_norm, assert, cut, trivial, revert, bool_congr, nat_congr,
        symmetry, transitivity, auto, split, left, right, autorewrite},
    morecomment=[s]{(*}{*)},
    showstringspaces=false,
    morestring=[b]",
    morestring=[d]’,
    extendedchars=false,
    sensitive=true,
    breaklines=false,
    basicstyle=\footnotesize\ttfamily,
    captionpos=b,
    columns=[l]flexible,
    identifierstyle={\ttfamily\color{black}},
    keywordstyle=[1]{\ttfamily\bfseries\color{dkviolet}},
    keywordstyle=[2]{\ttfamily\bfseries\color{dkgreen}},
    keywordstyle=[3]{\ttfamily\color{ltblue}},
    keywordstyle=[4]{\ttfamily\color{dkblue}},
    stringstyle=\ttfamily,
    commentstyle={\ttfamily\color{dkgreen}},
    literate=
    {\\forall}{{\color{dkgreen}{$\forall\;$}}}1
    {\\exists}{{$\exists\;$}}1
    {<-}{{$\leftarrow\;$}}1
    {=>}{{$\Rightarrow\;$}}1
    {==}{{\code{==}\;}}1
    {==>}{{\code{==>}\;}}1
    {->}{{$\rightarrow\;$}}1
    {<->}{{$\leftrightarrow\;$}}1
    {<==}{{$\leq\;$}}1
    {\\o}{{$\circ\;$}}1
    {\@}{{$\cdot$}}1
    {\/\\}{{$\wedge\;$}}1
    {\\\/}{{$\vee\;$}}1
    {++}{{\code{++}}}1
    {\@\@}{{$@$}}1
    {\\mapsto}{{$\mapsto\;$}}1
    {\\hline}{{\rule{\linewidth}{0.5pt}}}1
}[keywords,comments,strings]
\newcommand{\versionInfo}{{Version \today:\currenttime}}

\lstdefinestyle{CoqResourceIsolation}{
    language=Coq,
    morekeywords=[5]{step_running, machine_st, local_init_state, observe_output, job_done, next_job_rdy, output_job_resp, output_job_accept}, %
    keywordstyle=[5]{\ttfamily\color{dkred}}, %
    basicstyle=\scriptsize\ttfamily, %
    numbers=left,                %
    numberstyle=\scriptsize%
}

\lstdefinestyle{CoqEnclaveIsolation}{
    language=Coq,
    morekeywords=[5]{extract_dram, extract_rf, step_running, machine_st, enclave_params, can_start, wf_can_start, init_machine, init, spin_up_machine, merge_external_observations},
    keywordstyle=[5]{\ttfamily\color{dkred}}, %
    basicstyle=\scriptsize\ttfamily, %
    numbers=left,                %
    numberstyle=\scriptsize %
}

\newtheorem{theorem}{Theorem}

\crefname{definition}{definition}{definitions}
\Crefname{definition}{Definition}{Definitions}

\newif\ifincludenotes
\includenotesfalse

\newcommand{\note}[3]{}
\ifincludenotes
\renewcommand{\note}[3]{\textcolor{#1}{\textbf{[#2: #3]}}}
\fi

\newcommand{\Koika}{K\^oika}
\newcommand{\quartz}{Quartz}

\usepackage{soul}
\theoremstyle{definition}
\newtheorem{defn}{Definition}
\newcommand{\granite}{Granite}

\newcommand{\PageColorHook}{%
    \ifnum\value{page}>24
        \pagecolor{red!20} %
    \fi
}
\newcommand{\rf}{\mathit{rf}}
\newcommand{\pc}{\mathit{pc}}
\newcommand{\imem}{\mathit{imem}}
\newcommand{\dmem}{\mathit{dmem}}

\ifincludenotes
\AddEverypageHook{\PageColorHook}
\fi

\newcolumntype{C}{>{\centering\arraybackslash}p{1.4mm}}
\newlength{\bandside}

\date{}

\title{\granite{}: A Modular Methodology for Foundational Verification of Hardware-Software Leakage Contracts
}
\author{Stella Lau}
\affiliation{%
\institution{Massachusetts Institute of Technology}
  \country{USA}}

\author{Andres Erbsen}
\affiliation{%
  \institution{Google and University of Washington}
  \country{USA}}

\author{Adam Chlipala}
\affiliation{%
\institution{Massachusetts Institute of Technology}
  \country{USA}}
\renewcommand{\shortauthors}{}
\authorsaddresses{}
\begin{document}

\begin{abstract}
\granite{} is a methodology for modular verification of both functional correctness and nonleakage of RTL processors against ISA contracts.
We prove that the cycle-by-cycle timing of a pipelined RISC design---with speculation, precise interrupts, and I/O---is determined solely by observables specified in an ISA leakage contract.
For programs that keep observables independent of secrets (i.e., following the cryptographic-constant-time discipline), this result conclusively rules out information leakage through known and unknown timing side channels.

\granite{}'s specifications only constrain functional correctness and information-flow dependencies: not how many cycles an instruction takes, at which instruction an interrupt is handled, or the exact latencies of submodules such as multipliers and memory.
\granite{}'s central technique is \emph{leakage-aware refinement via determinism}, which establishes correctness and confidentiality together as trace equivalence with respect to a family of cycle-level, deterministic spec machines. 
Secret-independent nondeterminism is handled by \emph{existentially parameterizing} specifications with untrusted, deterministic functions acting only on public data.
Submodules are proved against their own leakage-aware specifications, and these proofs compose into the whole-design guarantee---which therefore holds over a space of secure implementations.

We believe this work is the first to achieve modular and foundational connection between instruction-set-level leakage contracts and microarchitecture-specific cycle-by-cycle execution with wire-level observations.
Our proofs compose with a certified static analysis that recognizes cryptographic-constant-time code to derive a single Rocq theorem about the cycle-by-cycle confidentiality of a hardware-and-software cryptographic implementation---eliminating every intermediate specification, including the ISA contract itself, from the trusted computing base.

\end{abstract}

\maketitle

\section{Introduction}
Microarchitectural timing attacks have revealed a systemic security lapse within computer architecture.
For years, cryptographers relied on constant-time programming to ensure that the wall-clock time when executing cryptographic software on hardware is independent of secrets.
However, starting with Spectre~\cite{spectre-oakland19} and Meltdown~\cite{meltdown-usenix18}, there has been a surge of microarchitectural timing vulnerabilities invalidating traditional constant-time assumptions and exposing novel leakage vectors.
Despite years of mitigations, new attack vectors continue to be discovered.

Formal verification against leakage-aware architectural specifications offers a compelling approach, adopted by this paper, for ruling out timing side channels.
A leakage-aware hardware-software contract specifies both the functional instruction set architecture (ISA) semantics and a declassification policy that characterizes exactly which architectural events may influence observable timing.
For example, the constant-time policy states that only addresses of memory accesses, branch conditions, and arguments to variable-latency instructions are leaked.
Proving that an RTL processor implementation satisfies such a contract yields a mathematical guarantee that hardware does not introduce unexpected timing channels beyond what is explicitly permitted by the contract.

For such a formal-verification approach to be trustworthy as the primary assurance mechanism and applicable to practical designs, the methodology should satisfy four critical requirements:
1) it is auditable, with a specification simple enough to inspect and with the RTL implementation removed from the trusted computing base (TCB);
2) it composes with software proofs for end-to-end, hardware-software guarantees that eliminate the hardware-software contract itself from the TCB;
3) it is expressive enough for real hardware, supporting the nondeterminism arising from asynchronous interrupts, inputs, and unspecified behaviour; and 
4) it is modular, both to tame the state-explosion problem encountered in hardware verification and to let designers reason about early-stage designs.
No prior methodology for verifying processors against hardware-software leakage contracts achieved all four.
\granite{} addresses this gap.

\paragraph{Existing approaches}
Various works propose sophisticated leakage models~\cite{guarnieri2021hardware,  ct-foundations, axiomatic-hw-sw-isca-22} and defenses~\cite{ghostMinion, SPT, prospect} but are either unverified or do not provide formal guarantees for RTL implementations.
A line of work~\cite{Leave,contractShadowLogic,upec-dit} uses model checkers to verify noninterference of RTL designs against ISA leakage models.
However, these methods assume functional correctness (and therefore require auditing the implementation), use a limited specification language of model checkers that does not smoothly compose with software proofs for end-to-end hardware-software guarantees, and lack the modular decomposition needed to tackle the state-explosion problem and scale to larger designs.
Another line of work~\cite{Fjfj,kami,lightbulb} uses proof assistants to prove functional correctness of processors as refinement against one-instruction-at-a-time (OIAAT) specifications, but refinement does not generally suffice for confidentiality as nondeterminism leaks secrets~\cite{interactive-information-flow-chong}.
Lastly, Notary~\cite{Notary} and Parfait~\cite{parfait} collapse the hardware-software stack to prove functional correctness and confidentiality for specific programs on specific processors, but the approaches used do not yield a blueprint for modular engineering and verification of their software and hardware components.

\subsection{Our Approach}
\granite{} is a methodology, formalized in Rocq, for end-to-end verification of both functional correctness and nonleakage of cycle-accurate RTL processors against ISA-level specifications. 
We show how to build modular proofs for processors and their components, with specifications that capture which information flows are allowed without prescribing specific implementation behaviours or requiring determinism.
Compared to existing hardware-verification techniques, the specification style we recommend is closer to that used in functional-correctness proofs~\cite{Fjfj,kami,lightbulb} than established noninterference-verification work~\cite{Leave,contractShadowLogic,upec-dit}. We present three aspects:

\paragraph{A specification approach capturing both functional correctness and nonleakage, supporting nondeterminism.}
\granite{} uses a specification approach, \emph{leakage-aware refinement via determinism}, that reduces both functional correctness and nonleakage to trace equivalence with respect to a family of deterministic, cycle-accurate state machines.
The central challenge is supporting nondeterminism without underspecifying (allowing unintended leakage) or overspecifying (ruling out secure implementations).
Our strategy is to \emph{existentially parameterize} each specification with a deterministic, public \emph{driver}---depending only on public information (public inputs, the program, and declassified leakage), responsible for resolving secret-independent nondeterminism such as the timing of outputting an MMIO request or at which instruction boundary to take a precise interrupt---and a deterministic \emph{witness} for nondeterminism that may legitimately depend on secrets. These parameters are untrusted.

We formalize nonleakage as the existence of a global \emph{leakage transformer} mapping the trace of public inputs to adversary observations (e.g.\ timing), with the (static) choice of that function left intentionally unconstrained to allow for implementation variability.
This property trivially implies the classical notion of noninterference.
A simple example appears in \Cref{fig:opener}.
This idea takes a variety of shapes when reconciling differences between specification and implementation; we show how to capture that:
\begin{itemize}
    \item A zero-skip multiplier's (public) latency can depend on whether either operand is zero but not on other aspects of input data.
    \item The response time of a memory subsystem may depend on addresses but not on values stored and loaded.
    \item While the instruction-set specification processes one instruction per step, a processor cycle may or may not commit a new instruction---but it must not decide based on secrets.
    \item Timing of external inputs and interrupts can influence processor progress, but values of inputs cannot in general be revealed.
    \item A processor can delay taking an interrupt (note that this decision can alter the sequence of instructions executed and thus the ISA-level observations!), but the decision of when to take the interrupt should not depend on secrets.
\end{itemize}

\begin{figure}
\small
\begin{verbatim}
Record req_t := { input_a : bv width; input_b: bv width }.
Record MulSt := { reqs: list req_t; hist : list ActionMethod }.
Definition leakage_of_AM (m: ActionMethod) : LeakEvent :=
  match m with
  | Enq req => LeakEnq ((req.(input_a) = 0) || (req.(input_b) = 0))
  | Deq => LeakDeq
  | Tick => LeakTick
  end.
Definition RespReadyOk impl spec := exists leakageTransformer,
  impl.(RespReady) = match spec.(reqs) with 
                     | [] => false 
                     | _ => leakageTransformer (map leakage_of_AM spec.(hist))
                     end.
\end{verbatim}
\caption{Example fragment of a module specification: the public timing of a zero-skip multiplier is a fixed function of whether its operands are zero. Functional correctness requires a pending request for \texttt{RespReady} to be high.
}
\label{fig:opener}
\end{figure}

We show how to formalize an ISA leakage contract (covering I/O, exceptions, and interrupts) in this style by lowering the ISA specification---written in the style of \textsc{RISCV-COQ}~\cite{riscv-coq,bedrock-ct} and extended with a leakage model emitting declassified events (e.g.\ memory addresses, branch conditions)---into a cycle-level machine whose nondeterminism is constrained by existential parameters. 
Functional correctness becomes I/O-trace equivalence over all programs, data, and inputs; nonleakage becomes the existence of a leakage transformer from public inputs, public data, and the specification leakage trace to public outputs.

\paragraph{A verification approach structured around horizontal and vertical modularity.}
\emph{Vertically} (c.f.~\Cref{fig:vertical}), a layered approach decouples ISA-specific details from microarchitectural design and parameterizes the proof over different ISA encodings and existential parameters;
we show that connecting the layer to RTL amounts to instantiating the existential parameters by running a public ``shadow'' copy of the implementation.
\emph{Horizontally}, we construct per-component specifications---with canonical representations capturing both correctness and leakage, plus sub-cycle invariants---extending prior one-method-at-a-time modular refinement~\cite{kami,Fjfj,bluespec-memocode04} to cycle-accurate nonleakage, yielding theorems general enough to cover classes of designs and to verify early-stage designs and defenses.

\paragraph{An end-to-end, machine-checked guarantee from software to RTL}
We build a confidentiality-and-correctness proof for a pipelined RISC processor (as synthesizable RTL) featuring branch prediction, exceptions, interrupts, and memory-mapped input and output.
We then demonstrate suitability of the instruction-set specification for integration verification by implementing and proving a static analysis for constant-time programming: any program accepted by the static analysis, run on a processor satisfying the specification, does not leak secrets at the RTL.
Instantiated on a Salsa20 program compiled with a standard RISC-V toolchain, this proof architecture yields a single, foundational proof connecting the source-level constant-time discipline to RTL and eliminates intermediate specification layers (including the hardware-software contract) from the trusted computing base.

\paragraph{Contributions.}
This paper contributes:
\begin{itemize}
    \item A specification approach for functional correctness and nonleakage under nondeterminism, via trace equivalence with deterministic, cycle-accurate state machines.
    \item A formal HW/SW leakage contract covering I/O, exceptions, and interrupts.
    \item A methodology for modular proofs of functional correctness and nonleakage of RTL designs.
    \item A machine-checked proof of a synthesizable, pipelined RISC processor (speculation, interrupts, I/O) against a constant-time ISA contract.
    \item An end-to-end proof combining the above with verified-constant-time software into an RTL confidentiality theorem, stated without reference to the novel specs.
\end{itemize}
Source code for \granite{} will be open-source and is included as an anonymized supplement.

\paragraph{Limitations and nongoals.}
The proofs we demonstrate rely on determinism of the hardware-description-language fragment we use, and thus it is not immediately clear how to extend our approach to reason about combined behaviour of modules in different clock domains.
Less fundamentally, we use a minimal RISC instruction set (assorted instructions from RISC-V), consider only single-hardware-thread executions with read-only instruction memory (we expect challenges with extensions on this front to be predominantly related to functional correctness rather than timing leakage), and do not actually connect our proofs to the compiler-verification work~\cite{bedrock-ct} that inspired aspects of our ISA specification.
The only side channel considered in this work is cycle-by-cycle timing; observations of within-cycle timing or power consumption are out-of-scope, as are attacks involving physical access.

\section{Background and Motivation}
\subsection{Hardware-Software Leakage Contracts}
The ISA is the contract between hardware and software.
An ISA such as RISC-V or x86 defines an architectural model, consisting of registers and memory, and how each instruction updates architectural state, granting hardware designers freedom to optimize (e.g.~via pipelining, branch prediction, and caching) so long as architectural state matches a sequential execution of instructions\footnote{Relaxed with weak memory, outside the scope of this paper.}. Traditionally, this contract has been purely functional, omitting microarchitectural state and abstracting over latency.

As the ISA says nothing about timing, optimizations can leak secrets through timing side channels.
A \emph{leakage contract} closes this gap by specifying which architectural events may influence timing, acting as a declassification policy.
Under the cryptographic constant-time policy, only control flow, addresses of memory accesses, and inputs to variable-latency instructions are leaked;
it is then the software programmer's responsibility to keep secrets out of these channels.
Numerous libraries hold all cryptographic code to this contract, either informally (BoringSSL, BearSSL) or formally (HACL$^*$~\cite{HACL} and EverCrypt~\cite{EverCrypt}).

Guarantees provided by these contracts are only as sound as the leakage model, and models routinely make assumptions that some real processors do not satisfy.
For example, HACL$^*$ assumes integer multiplication is constant-time, but integer multiplication can be variable-time on some ARM and i386 platforms; OpenSSL shipped cryptographic code deemed secure under a leakage model that did not account for time-variable arithmetic operations~\cite{ammanaghatta2022enforcing}; and hardware verification exposes still more (c.f.\ \Cref{sec:case_study}).

\paragraph{Transient-execution attacks}
Worse, transient-execution attacks such as Spectre~\cite{spectre-oakland19} and Meltdown~\cite{meltdown-usenix18} revealed that modern processors violate even the constant-time contract: speculatively executed instructions leave secret-dependent side effects in caches, TLBs, and branch predictors the policy ignores. 
Proposed defenses~\cite{delay-on-miss, invisispec, safespec, muontrap, ghostMinion} aim to restore the guarantee discussed next, but none has a verified RTL implementation.

\subsection{Contractual Noninterference}
\granite{}'s formalization of nonleakage is based on \emph{speculative noninterference}~\cite{STT, SPT, pensieve, contractShadowLogic, ct-foundations}, which (informally) captures the notion that if a program does not leak information under the ISA model (e.g.\ it satisfies the constant-time contract), then running on hardware does not leak information.
This property is the one we verify; because leakage can arise without speculation and contracts generalize beyond it, we adopt the more general term contractual noninterference.
We first state the property in the deterministic setting (consistent with definitions from prior work on speculative noninterference), where the architectural specification and implementation can both be modeled as deterministic state machines; \Cref{sec:background:nondeterminism} discusses challenges with generalizing to nondeterminism; and \Cref{sec:overview:nondeterminism} presents our approach to generalize to nondeterministic specs. 

\begin{defn}[Contractual noninterference without nondeterminism]\label{defi:noninterference}
Let $\mathit{Pub}$ represent public initial state (such as program instructions and public data), and $\mathit{Sec^1}$ and $\mathit{Sec^2}$ represent secret data. 
Let $O_{\mathit{ISA}}^n(\mathit{Pub}, \mathit{Sec})$ denote the specification leakage trace generated by running the ISA machine for $n$ steps and $O_{\mu}^m(\mathit{Pub}, \mathit{Sec})$ denote the trace of cycle-accurate observations generated by running the RTL implementation for $m$ clock cycles. 
The implementation satisfies contractual noninterference without nondeterminism if $\forall \mathit{Pub}$:%
\begin{align*}%
& \left(\forall \mathit{Sec}^1, \mathit{Sec}^2, n.\,O^n_{\mathit{ISA}}(\mathit{Pub}, \mathit{Sec}^1) = O^n_{\mathit{ISA}}(\mathit{Pub}, \mathit{Sec}^2)\right)  \rightarrow  \\
& \left(\forall \mathit{Sec}^1, \mathit{Sec}^2, m.\,O^m_{\mu}(\mathit{Pub}, \mathit{Sec}^1) = O^m_{\mu}(\mathit{Pub}, \mathit{Sec}^2)\right)  
\end{align*}
\end{defn}

\paragraph{Leakage transformers.}
The above definition is a 4-copy property, stating that if the ISA leakage trace is independent of secrets (or therefore, the ISA leakage trace is public information), then the microarchitectural leakage trace is also independent of secrets (and therefore, there exists a function such that the microarchitectural trace is a function of public information).
The latter function is exactly the ``leakage transformer'' we use to express nonleakage:
\begin{align*}%
 & \left( \exists g.\,\forall n, \mathit{Sec}.\, O^n_{\mathit{ISA}}(\mathit{Pub}, \mathit{Sec}) = g(n, \mathit{Pub}) \right) \rightarrow \\
& \left( \exists f.\,\forall m, \mathit{Sec}.\,O^m_{\mu}(\mathit{Pub}, \mathit{Sec}) = f(m, \mathit{Pub}) \right).
\end{align*}

\paragraph{Challenges with proving contractual noninterference.}
Exploiting the assumption that the ISA leakage trace is public requires relating microarchitectural state to instruction-set-level state (which instructions have retired and \emph{will} retire), a challenging task akin to proving functional correctness.
An important simplification in model-checking-based work~\cite{Leave,contractShadowLogic,upec-dit} assumes functional correctness of the processor, using the retire stage of the processor in place of the state of the ISA execution.
While appealing due to avoiding the need to model and reason about the instruction-set specification, this approach runs into state-space-explosion challenges (due to needing to show that an instruction will retire).
Furthermore, we are not aware of any work that establishes an integrated correctness-and-confidentiality theorem based on separate functional verification and constant-time-assuming-functional verification of a processor.
Doing so seems challenging given leading approaches for comprehensive functional verification allow a rather flexible ``flushing relation''~\cite{su1996automatic,Fjfj} between the retire-stage state and instruction-set-level state instead of establishing strict equality.
As a goal of this work is a combined hardware-software result in which the hardware-software contract is only an intermediate specification, we seek a different approach.

\subsection{Refinement and Nondeterminism}\label{sec:background:nondeterminism}
Our approach is to extend---to the nonleakage setting---an approach to processor-functional-correctness verification~\cite{kami, Fjfj} based on using proof assistants to establish correctness via refinement: i.e., every behaviour of the implementation is a possible behaviour of the specification.
Refinement extends cleanly to nonleakage for deterministic specifications (as trace equivalence proved via bisimulation), but architectural specifications support various types of nondeterminism to allow for interaction with the external world and to allow hardware designers freedom to optimize. 
Sources of nondeterminism, covered by this paper\footnote{We elide, for example, nondeterminism from weak memory but discuss how the framework could be extended in~\Cref{sec:eval:mem}.}, include:
\begin{itemize}
    \item Execution time: the number of cycles it takes to execute instructions;
    \item Input nondeterminism from the external world, such as from MMIO and interrupts;
    \item Secret-independent nondeterminism, such as at which instruction boundary an asynchronous interrupt is handled; and
    \item Unconstrained output nondeterminism, such as the value on a \texttt{data} signal when the corresponding \texttt{valid} bit is low in a ready/valid handshaking interface.\footnote{Here, we assume it is the consumer's responsibility to use the \texttt{data} signal only when the corresponding \texttt{valid} bit is high.} 
\end{itemize}

\subsubsection{Specification challenges arising from inputs and nondeterminism}\label{sec:nondeterminismChallenges}
Nondeterminism gives rise to challenges with straightforward application of refinement techniques to \Cref{defi:noninterference}:

\paragraph{Challenge 1: Nondeterminism leaks secrets: the need for secret-independent nondeterminism}
Classical refinement becomes unsound for nonleakage: when a specification nondeterministically permits several behaviours, an implementation may resolve the choice on secrets---e.g., taking an interrupt immediately if and only if a secret bit is set---and the leak is then ``explained away'' by an adversarial choice of specification nondeterminism.
To preserve microarchitectural flexibility while preventing side channels, specification frameworks must support secret-independent nondeterminism to ensure that while an implementation retains the freedom to resolve design choices (such as at which instruction boundary to take an interrupt), the choice is independent of secrets.

\paragraph{Challenge 2: Nondeterminism from inputs and interrupts affects the specification leakage trace.}
If execution time were the sole source of nondeterminism and the software leakage trace were independent of execution time (as is standard in idealized constant-time models~\cite{bedrock-ct}), then \Cref{defi:noninterference} would be sufficient, as the specification leakage trace would be invariant under implementation design choices\footnote{This definition is used in prior work~\cite{bedrock-ct}, which does not address input nondeterminism.}. 
Real-world hardware, however, introduces dynamic sources of nondeterminism that directly influence the specification's leakage trace.
For example, the values returned by MMIO requests can be affected by previous interactions with the external world, which has downstream effects on leakage.
Additionally, asynchronous interrupts are scheduled nondeterministically:
while an ISA may mandate that precise interrupts occur at instruction boundaries~\cite{riscv-volume2}, it often permits flexibility regarding which specific instruction boundary traps the event.
This hardware-driven scheduling decision actively diverges program control flow to an interrupt handler, altering the resulting sequence of specification leakage events.

\paragraph{Challenge 3: Inputs and different step sizes}
The mismatch in step sizes between instruction-level software specifications and cycle-level hardware implementations not only is a key cause of the underspecification leading to microarchitectural timing side channels, but it also complicates interactions with the external world\footnote{This step-size mismatch is also a key challenge in verifying functional correctness, discussed in \Cref{sec:overview_proof_strategy}.}.
The external world can be modeled as a function of the cycle-level outputs of the microarchitectural machine, but high-level ISA specifications typically lack structural notions of clock cycles.
Thus, \granite{} lowers the ISA machine into a cycle-level state-transition system so both layers interact with a synchronized model of the external world, enabling a reduction of functional correctness and nonleakage to trace equivalence provable by standard single-cycle induction.

\subsection{Specifying Information Flow in Combinational Methods}
We augment the specifications used for modular functional-correctness verification of the processor's submodules with timing specifications.
The intuition is that any functional-correctness specifications suitable for modular verification of the processor must already capture some relation between the inputs and outputs of each module that is sufficiently precise to eventually relate these values to instruction-set-level state.
Augmenting these specifications with constraints on what the timing of the output may depend on \emph{in terms of the module's inputs} allows confidentiality verification to benefit from the modular structure of functional-correctness proofs.
While functional-correctness proofs are nontrivial, the methodology for crafting them is established and reasonably consistent across different hardware-description languages~\cite{kami,Fjfj,VerilogSemanticsJoonwon}.

As noted (\Cref{sec:nondeterminismChallenges}), the nondeterminism and refinement that functional-correctness frameworks embrace do not suffice for confidentiality.
Our approach, therefore, is to show equivalence (degenerately, refinement of a deterministic specification) after resolving all nondeterministic choices preemptively by instantiating the existential parameters.

\subsubsection{\quartz{}: A deterministic HDL enabling one-method-at-a-time semantics.}\label{sec:quartz:intro}
This strategy in turn leads us to choose a deterministic language, \quartz{}, for expressing our hardware designs.
We work with the intersection of Rocq's native language Gallina and SystemVerilog functions operating on pairs, sums, records, and arithmetic modulo powers of 2.
All concurrency is implicit---independent subexpressions in a functional program can be evaluated in parallel, and hardware-synthesis tools naturally produce implementations that do so.
Strictly speaking, all operations syntactically representable in our fragment are combinational. 
As RTL model-checking tools can trivially prove equivalence of encodings of the same circuit as different RTL programs, and our fragment can express any circuit, this encoding choice is not a capability limitation.

Following the tradition of Bluespec and hardware-verification frameworks it inspired~\cite{kami,koika-pldi20,Fjfj}, we specify modules in terms of atomic \emph{methods} that deterministically produce updated state and output given a starting state and input, perhaps calling other methods in the process.
Each clock cycle involves several methods firing, in our case through deterministic, purely functional calls from the top-level \texttt{tick} method.
A module implementation can then be verified independently as trace equivalence against a deterministic spec, existentially instantiated.

\section{Overview}
\granite{} establishes an end-to-end guarantee---a constant-time program running on synthesizable RTL leaks nothing through timing---by composing a chain of proofs, shown in \Cref{fig:vertical}.
This section highlights aspects of the proof chain and introduces an example, a zero-skip multiplier, that exposes the specification challenge of the paper: a specification must admit a whole family of functionally correct designs with legitimately varying timing yet forbid any timing that depends on secrets.

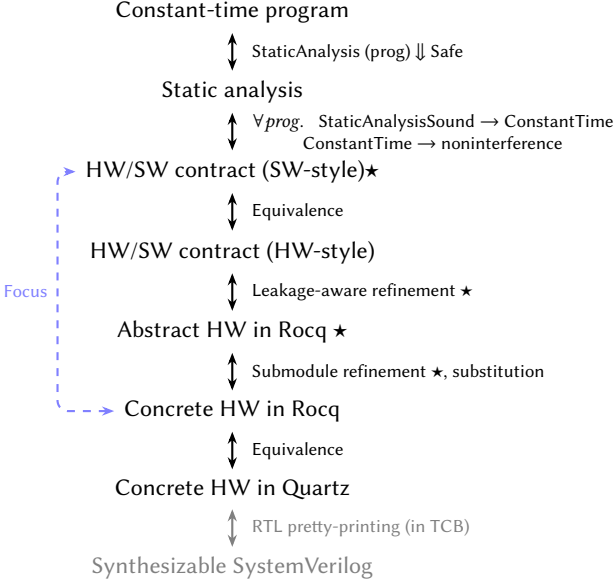
\begin{figure}
\centering
\begin{minipage}[t]{0.6\textwidth}
\centering
\begin{tikzpicture}[baseline=(current bounding box.north),
    block/.style = {
        rectangle, 
        thick, 
        fill=white,
        align=center,
        font=\sffamily\small %
    },
    arrow/.style = {
        {Stealth[scale=0.7]}-{Stealth[scale=0.7]}, %
        thick,
    },
    labelstyle/.style = {
        font=\scriptsize\sffamily,
        fill=white, 
        inner sep=2pt
    }
]
    \node[block] (A) {Constant-time program};
    \node[block] (B) [below=0.5cm of A] {Static analysis};
    \node[block] (C) [below=0.5cm of B] {HW/SW contract (SW-style)$\star$};
    \node[block] (D) [below=0.5cm of C] {HW/SW contract (HW-style)};
    \node[block] (E) [below=0.5cm of D] {Abstract HW in Rocq $\star$};
    \node[block] (F) [below=0.5cm of E] {Concrete HW in Rocq};
    \node[block] (G) [below=0.5cm of F] {Concrete HW in Quartz};
    \node[block, text=gray] (H) [below=0.5cm of G] {Synthesizable SystemVerilog};

    \draw [arrow] (A) -- (B) node[midway, anchor=west, right=5pt, labelstyle] {StaticAnalysis (prog) $\Downarrow$ Safe};
    \draw [arrow] (B) -- (C) node[midway, anchor=west, right=5pt, labelstyle, align=center] {$\forall \mathit{prog}$. \; StaticAnalysisSound $\rightarrow$ ConstantTime \\
                    ConstantTime $\rightarrow$ noninterference};
    \draw [arrow] (C) -- (D) node[midway, anchor=west, right=5pt, labelstyle] {Equivalence};
    \draw [arrow] (D) -- (E) node[midway, anchor=west, right=5pt, labelstyle] {Leakage-aware refinement $\star$};
    \draw [arrow] (E) -- (F) node[midway, anchor=west, right=5pt, labelstyle] {Submodule refinement $\star$, substitution};
    \draw [arrow] (F) -- (G) node[midway, anchor=west, right=5pt, labelstyle] {Equivalence};
    \draw [arrow, color=gray] (G) -- (H) node[midway, anchor=west, right=5pt, labelstyle, text=gray] {RTL pretty-printing (in TCB)};
    \draw [arrow, rounded corners, dashed, blue!50] (F.west) -- ++(-.75,0) |- (C.west) 
        node[pos=0.25, left, font=\scriptsize\sffamily] {Focus};
\end{tikzpicture}
\end{minipage}\hfill
\begin{minipage}[t]{0.39\textwidth}%
\caption[Vertical modularity]{Layered, end-to-end proof structure.
A program is deemed constant-time by a static analysis proven sound against a software-style ISA contract.
Observations are the timing of externally observable behaviours (e.g.~MMIO output).
An equivalence proof lowers the software-style ISA contract to a hardware-style ISA contract that uses hardware-optimized, combinational decode and execute logic shared across instructions.
A processor implementation with abstract submodules, specified using leakage contracts, is proven to refine the hardware-style ISA contract. 
After proving concrete submodules individually satisfy their specifications, the resulting processor is proven equivalent to an implementation in \quartz{} and pretty-printed to SystemVerilog.
}\label{fig:vertical} 
\end{minipage}
\end{figure}

\subsection{Example: Zero-Skip Multiplier}\label{sec:overview:example}
Before the full ISA contract, we introduce the specification style on two toy examples.
We provide intuition for how a family of deterministic specifications can admit a space of functionally correct designs with secure timing variations.

\paragraph{Warm-up: a combinational circuit.}
The smallest interesting case has no state.
A combinational circuit $\sigma$ maps a public input and a secret input to a public output and a secret output, $(out_{\mathit{pub}}, out_{\mathit{sec}}) = \sigma(in_{\mathit{pub}}, in_{\mathit{sec}})$.
Confidentiality requires only that the public output not depend on the secret input: that there exist a function $f$ with $out_{\mathit{pub}} = f(in_{\mathit{pub}})$, while $out_{\mathit{sec}}$ is free to be any function of everything.
We write this requirement as
$$ \exists f, g.\;\forall in_{\mathit{pub}}, in_{\mathit{sec}}.\;
   \sigma(in_{\mathit{pub}}, in_{\mathit{sec}}) = \bigl(f(in_{\mathit{pub}}),\, g(in_{\mathit{pub}}, in_{\mathit{sec}})\bigr). $$
The exact value of $out_{\mathit{pub}}$ is left unspecified, but the existential $f$ forces it to be secret-\emph{independent}: this device is our basic one for capturing secret-independent nondeterminism without overconstraining $\sigma$.
The unconstrained $g$ captures behaviour that is permitted to depend on secrets.

\paragraph{Adding state and time: the zero-skip multiplier.}
The combinational case fixed \emph{which} outputs may depend on secrets.
State and clock cycles raise a second question: \emph{when} does an output appear, and what may that timing depend on?
We answer both with the same strategy as before but applied over traces rather than single values: every public output, now including the cycle when it becomes available, must be a fixed function of the module's public inputs trace.

Consider a zero-skip multiplier that computes $a * b$, returning in one cycle when either operand is zero and in $n$ cycles otherwise.
We intend to place it inside a processor whose multiply instruction declassifies only \emph{whether} an operand is zero; the multiplier's contract must therefore permit the zero-skip optimization while forbidding any other data-dependent timing.
We model implementation and specification as Mealy machines\footnote{A Mealy machine is a \emph{deterministic} finite-state machine whose output values are determined both by its current state and the current input.} over the method calls \texttt{Enq}, \texttt{Deq}, \texttt{Tick}, \texttt{Full}, \texttt{RespReady}, and \texttt{Peek}, with correctness defined as output-trace equivalence, summarized in \Cref{fig:mult:method}.

A specification should admit a whole family of secure implementations, so it leaves several behaviours nondeterministic---of the two kinds from the warm-up.
A secret-independent \emph{driver} (here \texttt{d.RespReady}/\texttt{d.Full}, a function of only the public leakage trace) fixes the timing, seeing the zero bit but nothing else, playing the role of $f$.
An unconstrained \emph{witness} (here the \texttt{Peek}-while-invalid value, which may reflect secrets) plays the role of $g$.
These parameters are untrusted: an implementation is functionally correct and secure exactly when \emph{there exist} driver and witness instantiations under which its trace equals the specification's.
\Cref{sec:multiplier} gives the full contract in Rocq and proof method---we show that the witness is trivially instantiated with the implementation, and the driver is instantiated with a ``shadow copy'' of the implementation with inputs reconstructed from public inputs.

\begin{figure}
\includegraphics[width=\linewidth]{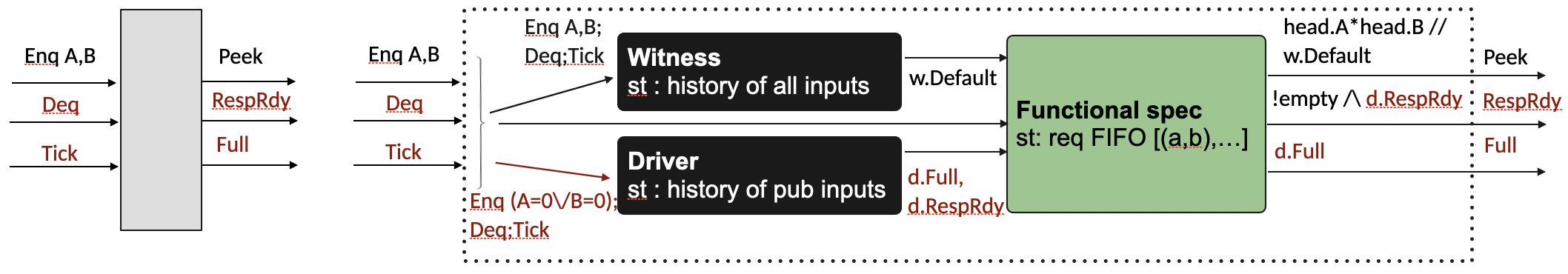}
\caption{Implementation (left) and specification (right) of a pipelined, zero-skip multiplier. 
Functional correctness is specified based on a canonical representation of state as a FIFO of requests.
Leakage-aware nondeterminism is specified using witness and driver components, with canonical representation of state as the history of their inputs.
An implementation is correct and secure if there exist a driver (a function of the history of public inputs; deciding latency) and witness (a function of all inputs; deciding the value of \texttt{Peek}-when-not-ready) such that it is equivalent to the specification instantiated with the parameters.
}
\label{fig:mult:method}
\end{figure}

\subsection{Security Goal and Threat Model}\label{sec:threatModel}
We consider HW/SW contracts based on three components: a software or ISA specification with leakage semantics, a hardware or microarchitectural implementation, and an adversary model specifying what aspects of the microarchitecture are observable.

\paragraph{Adversary capability}
\granite{} considers adversaries with direct access to the microarchitecture's digital I/O boundary, capable of driving input signals and observing output signals at every clock cycle.
We formalize this capability by defining a public observation function $O_{\mu}$ over execution traces of the microarchitectural state machine.
For example, our case study (\Cref{sec:case_study}) defines observations as the cycle-by-cycle ready/valid handshake wires of the external MMIO bus.
Fixing observations at this system boundary means internal microarchitectural transitions need not be audited directly.\footnote{Our intermediate theorems also prove secret independence of other commonly used observations, such as instruction completion time or the timing and addresses of all memory requests.}
The threat model includes digital timing side channels but excludes arbitrary side channels such as power or EM radiation.

\paragraph{Security goal}
\granite{} proves the implementation \emph{satisfies its leakage contract}---the contractual-noninterference property of \Cref{defi:noninterference}, so observations are a function of what the contract declassifies or designates as public---and additionally proves \emph{functional correctness} (refinement of the ISA's I/O behaviour, universally over all programs, data, and inputs). 
\Cref{sec:overview:nondeterminism} generalizes contractual noninterference to nondeterminism and I/O over infinite traces.

\paragraph{ISA specification with leakage semantics}
\granite{}'s methodology is parameterized over an ISA specification defined as a state machine, consisting of architectural states and a transition function, together with a \emph{leakage function} that emits the information the contract declassifies at each step.
We focus on ISAs with sequential, one-instruction-at-a-time semantics on a single core.\footnote{See \Cref{sec:eval:mem} for discussion of generalizing beyond sequential ISAs and to multicore.}
As a concrete example, our case study considers a representative subset of RISC-V under the constant-time policy.
Architectural states are $\sigma = \langle \rf, \pc, \imem, \dmem \rangle$---a register file $\rf$, program counter $\pc$, and separate instruction and data memories $\imem$ and $\dmem$ mapping addresses to bytes (we restrict to non-self-modifying programs)---with $\mathit{step}: \sigma \rightarrow \sigma$ giving the standard ``one-instruction-at-a-time'' fetch-decode-execute transition.
A leakage function $\mathit{leak}$ computes the leakage of an instruction, excerpted below, for a constant-time contract for RISC-V with zero-skip multiplication. 
We additionally treat the byte-encoded instructions themselves as leaked, along with writes to interrupt-related CSRs (c.f.\ \Cref{sec:case_study} explains why).
\begin{lstlisting}[numbers=none,label=lst:leak]
| Add rd rs1 rs2 => LeakAdd (* All inputs are secret. *)
| Beq rs1 rs2 offset => LeakBeq (rf[rs1] = rf[rs2]) (* The branch condition is leaked. *)
| Lw rd rs1 offset => LeakLw (rf[rs1]) (* The load address is leaked. *)
| Sw rs1 rs2 offset => LeakSw (rf[rs1]) (* The store address is leaked (but the store data is not). *)
| Mul rd rs1 rs2 => LeakMul (rf[rs1] = 0 \/ rf[rs2] = 0) (* Whether either argument is zero is leaked. *) 
| Csrrw rd rs1 csr => match csr with
                      | mtvec => LeakCsrrwMtvec rf[rs1]
                      | mie => LeakCsrrwMie rf[rs1]
                      | _ => LeakCsrrw 
                      end
\end{lstlisting}

\subsection{Architectural Specifications Under Nondeterminism}\label{sec:overview:nondeterminism}
Recall the three challenges of \Cref{sec:nondeterminismChallenges}: secure implementation choices affect the leakage trace, nondeterminism can leak secrets, and step sizes mismatch. We resolve all three with a single construction, applying the secret-independent driver and secret-dependent witness that determinized the multiplier of \Cref{sec:overview:example} to the ISA machine.

First, we lower the one-instruction-at-a-time ISA machine into a cycle-accurate state machine, aligning specification steps with implementation cycles and letting both interact with a synchronized model of the external world.
Second, we \emph{determinize} the remaining nondeterminism with explicit existential parameters~\cite{mtisolation}, of the same two kinds as in the multiplier: 
\begin{itemize}
    \item a secret-independent \emph{driver}, fed only public data, that resolves secret-independent choices (i.e.\ when the machine steps and triggers I/O, or when it takes a pending interrupt); and 
    \item an unconstrained \emph{witness}, which may depend on secrets, modeling unspecified behaviour such as the value driven on a data bus while its \texttt{valid} bit is low.
\end{itemize}
A separate secret-independent \emph{leakage transformer} maps the specification's declassified leakage trace to the adversary's cycle-level observations; its existence \emph{is} the security property.
An implementation is then functionally correct and secure when there \emph{exist} a driver, witness, and leakage transformer under which it is trace-equivalent to the resulting deterministic, cycle-accurate specification (\Cref{fig:spec_overview}).
All three parameters are untrusted, so only the cycle-accurate ISA machine---not the implementation nor the parameters---must be audited.
\Cref{sec:spec:hwsw} develops the construction in full, and \Cref{sec:integration} shows that trace equivalence to it implies the classical noninterference property of \Cref{defi:noninterference}.

\begin{figure}
\begin{subfigure}{0.47\textwidth}
    \centering
    \includegraphics[height=3cm]{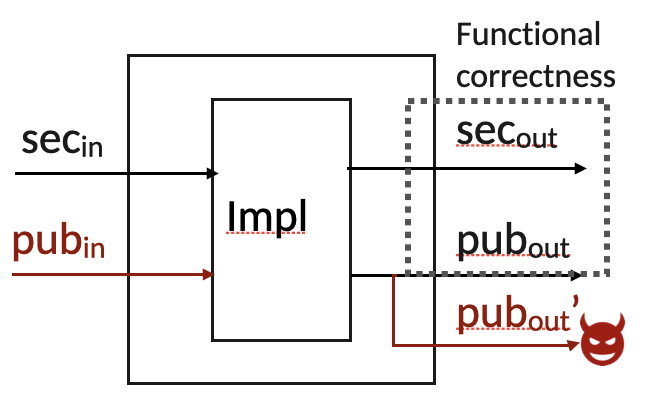}
    \caption{Implementation state machine, cycle-level, with public and secret I/O at the wire level. The public output wires, observed by the adversary, are duplicated for convenience to assert both functional correctness and nonleakage as trace equivalence.}
\end{subfigure}\hfill
\begin{subfigure}{0.52\textwidth}
    \centering
    \includegraphics[height=3cm]{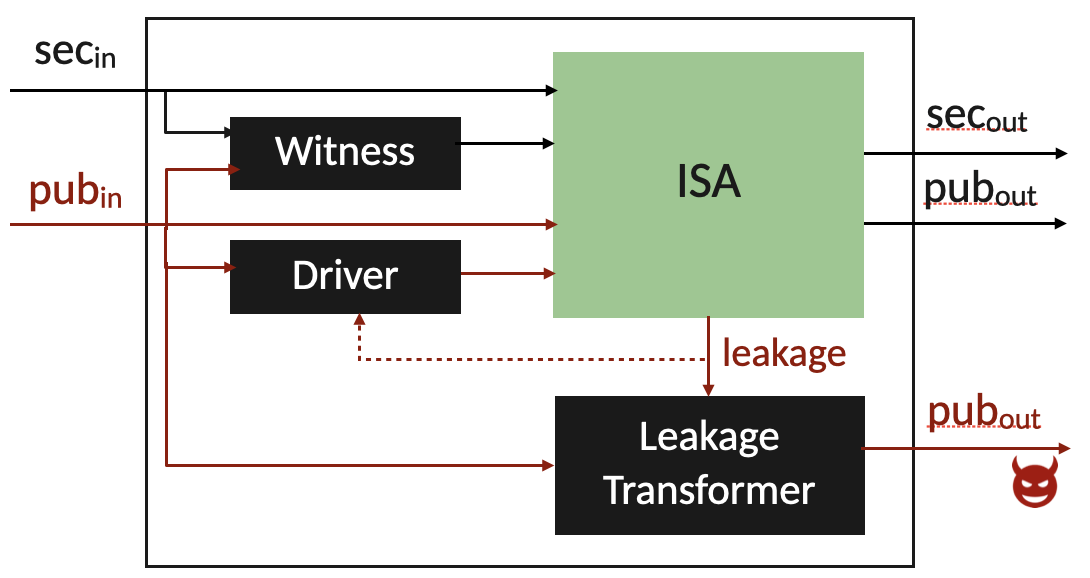}
    \caption{Specification state machine, cycle-level, constructed by parameterizing a trusted ISA machine generating declassified leakages with a secret-independent driver, a secret-independent leakage transformer, and an unconstrained witness. The black boxes are untrusted.}
\end{subfigure}
\caption[Cycle-level specification state machine]{An implementation satisfies a leakage-aware ISA contract---and hence is both functionally correct and secure---if there exist a witness, driver, and leakage transformer under which the implementation and specification traces are equivalent.
Red wires carry public, adversary-observable information.
Only public information is fed into the leakage transformer, which produces the adversary-observable output.}
\label{fig:spec_overview}
\end{figure}

\subsection{Two Axes of Modularity}\label{sec:overview:modularity}
\granite{} decomposes the proof along two axes, which together tame the state-space explosion of monolithic hardware verification and let the framework reason about early-stage designs.

\paragraph{Vertical Modularity.}
\Cref{fig:vertical} shows \granite{}'s layered refinement chain: a constant-time program, a \emph{software-style} ISA contract amenable to software proofs (\Cref{sec:case_study}), a \emph{hardware-style} contract sharing combinational decode/execute/writeback logic (reused across spec and implementation), abstract-submodule and HDL-independent layers (useful for reasoning about early-stage designs and defenses), and finally \quartz{} with a trusted SystemVerilog pretty-printer.
Crucially, refinement preserves not just functional state but leakage: we chain leakage-transformer functions so the implementation's cycle-level leakage trace is provably a function of the top-level specification's.

\paragraph{Horizontal modularity (\Cref{sec:modular}).}
Horizontally, the abstract-hardware design decomposes into submodules---FIFOs, multipliers, branch predictors, memory---that interact only through method calls and are each specified (correctness and nonleakage) and verified independently (as previewed in \Cref{sec:overview:example}). 
This tames state-space explosion and lets a designer swap algorithms or reason about an early-stage defense without redoing the top-level proof; \Cref{sec:modular} develops the details.

\section{Specifying HW/SW Leakage Contracts}\label{sec:spec:hwsw}
This section shows how \granite{} turns a traditional ISA contract into a family of deterministic, cycle-accurate state machines suitable for leakage-aware refinement proofs.
We build up to, and present details of, the specification in \Cref{fig:spec_overview} supporting I/O and interrupts.
Starting with an example baseline ISA state machine, we lower it to a cycle-accurate machine that reconciles the step-size mismatch with implementations (\Cref{sec:spec:baseline}), describe its existential parameters (\Cref{sec:existential_parameters}), and discharge both functional correctness and nonleakage with a single trace-equivalence obligation (\Cref{sec:spec:refinement}).
We then discuss the bug classes ruled out and outline the proof strategy.

\subsection{Lowering an ISA Contract to a Cycle-Accurate Machine}\label{sec:spec:baseline}

\paragraph{State machines.}
We model state machines as Mealy machines. 
A \emph{state machine} $M : \mathit{Machine}\ I\ O$ is a
triple $(\mathit{State}, s_0, \delta)$ of a state type, an initial state
$s_0 \in \mathit{State}$, and a transition function $\delta : \mathit{State} \times I \rightarrow O \times \mathit{State}$. 
Driving $M$ with an input sequence $\vec{\imath_n} = i_1 \cdots i_n$ yields
an $n$-element output trace, written $\mathit{Tr}_M^n(\vec{\imath})$.
Both the specifications and implementations are Mealy machines; this section works toward equating their traces.

\paragraph{Baseline ISA machine.}
The baseline ISA is a Mealy machine over architectural states
$\sigma = \langle \rf, \pc, \imem, \dmem \rangle$---a register file, program counter, and separate instruction and data memories (\Cref{sec:threatModel}). 
Instructions are encoded as an inductive type $\mathit{Instr}$ with a decode function $\mathit{decode}$ mapping machine words to $\mathit{Instr}$ and an execute function $\mathit{exec} : \sigma \times \mathit{Instr} \rightarrow \sigma \times \mathit{LeakEvent}$ that updates the architectural state and emits the per-instruction leakage event prescribed by the contract.
The baseline one-instruction-at-a-time machine, which takes no inputs, has the following state-transition function:
$$ \mathit{step} : \sigma \times \mathit{unit} \rightarrow \sigma \times \mathit{LeakEvent}
   \triangleq \lambda s.\ \mathit{exec}(s, \mathit{decode}(\imem[s.\pc])).$$

\paragraph{Supporting I/O with a wire-level, cycle-level machine}
We want to support I/O (we use MMIO as a running example).
As such, we aim to relate specifications and implementation \emph{cycle-by-cycle}, so that both react to the same external inputs at the same instant.
However, $\mathit{step}$ executes a whole instruction atomically, whereas an implementation may spread one instruction across many cycles and retire nothing at all on a stalled cycle.
In addition, an MMIO load request can block until a response is received, breaking the one-instruction-at-a-time semantics.
Furthermore, we need to add a trusted wire-level interface to support MMIO: to demonstrate, we extend $\sigma$ with a handshaking ready/valid/data interface, considering ready/valid signals to be public and data signals to be secret.

\paragraph{Cycle-accurate machine with I/O}
We lower the ISA into a cycle-accurate machine of type
$$ \mathit{Cycle{-}ISA} : \mathit{Machine}\ (\mathit{PubIn} \times \mathit{SecIn} \times \mathit{DriverOut} \times \mathit{WitnessOut})\ (\mathit{PubOut} \times \mathit{SecOut} \times \mathit{LeakEvent}) $$
whose transitions are split so that each carries at most one distinct leakage or observable effect, letting a single implementation cycle map to zero or more specification steps.
As shown in \Cref{fig:spec_overview}, the machine is existentially parameterized with a driver and witness:
\begin{align*}
\mathit{Driver} &: \mathit{Machine}\ (\mathit{PubIn} \times \mathit{LeakEvent})\ \mathit{DriverOut}  \\
 \mathit{Witness} &: \mathit{Machine}\ (\mathit{PubIn} \times \mathit{SecIn})\ \mathit{WitnessOut} 
\end{align*}
responsible for resolving secret-independent choices (in this case, whether to take a step and cause external I/O observations from sending/receiving MMIO requests) and plausibly-secret-dependent choices (e.g.\ the value of the data wire when the valid signal is low) respectively.
At heart this machine is still the trusted OIAAT ISA, now emitting its per-instruction leakage stepwise (under the constant-time policy: branch conditions, load/store addresses, and inputs to variable-latency instructions);
it is the only component that must be audited; the driver and witness are untrusted.

Here, $ \mathit{Cycle{-}ISA} $ is decomposed into three stages (extended in \Cref{fig:spec_state_machine} with interrupts) in which MMIO is the only observable effect:
\begin{enumerate}
    \item \texttt{StepLeak}: when driven, emits the leakage of the next instruction (at the program counter) without executing it.
    \item \texttt{StepInstr}: when driven, runs one fetch--decode--execute to completion---except when the instruction issues an MMIO request, in which case it emits the observable request (by appropriately setting ready/valid/data wires) and waits.
    \item \texttt{StepWaitMMIOResp}: when driven, consumes an MMIO response if one is available (else blocks). \emph{Which cycle} the response arrives is observable (altering the \texttt{ready} signal) and may differ from when the external world sent it.
\end{enumerate}

\paragraph{The wire-level interface and functional correctness.}
A trusted layer bridges the gap between the specification's abstract ``send an MMIO request'' with the wire-level ``ready/valid'' I/O interface: the outgoing \texttt{valid} is held high while a request is pending, and the response is consumed when the input \texttt{ready} is also high; similarly, the outgoing \texttt{ready} is held high exactly when the machine is in \texttt{StepWaitMMIOResp}.
This interface specifies functional correctness: the ISA spec enforces that outgoing requests correspond to MMIO requests in the ISA semantics.

\paragraph{Supporting interrupts}
In the above, an implementation has freedom to vary timing (in a secret-independent manner) on the MMIO interface and to set the MMIO data wire when the valid signal is low; we additionally want to support interrupts.
Typically, interrupts must be precise, corresponding to a sequential execution of the program that can take an interrupt after an instruction, but the implementation can choose after which instruction to take an interrupt.
This choice should be independent of secrets but is a true nondeterministic choice that depends on I/O and alters the instructions being executed.

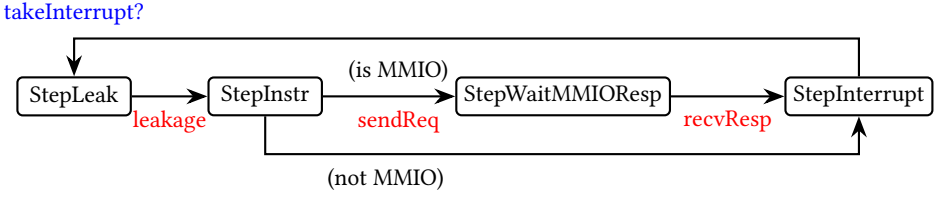
\begin{figure}
\small
\begin{tikzpicture}[
    block/.style={
        draw, 
        rectangle, 
        minimum width=1.5cm, 
        minimum height=0.5cm, 
        align=center,
        rounded corners=2pt,
        thick
    },
    arrow/.style={
        -{Stealth[scale=1.2]},
        thick
    },
    label style/.style={
        font=\small, 
        inner sep=5pt
    }
]
    \node[block] (leak) {StepLeak};
    \node[block, right=1cm of leak] (instr) {StepInstr};
    \node[block, right=1.75cm of instr] (mmio) {StepWaitMMIOResp};
    \node[block, right=1.5cm of mmio] (interrupt) {StepInterrupt};

    \draw[arrow] (leak) -- (instr) 
        node[pos=0.5, below, label style] {\textcolor{red}{leakage}};

    \draw[arrow] (instr.east) -- ++(1,0) |- (mmio.west)
        node[pos=0.25, above, label style] {(is MMIO)}
        node[pos=0.25, below, label style] {\textcolor{red}{sendReq}};
        
    \draw[arrow] (instr.south) -- ++(0,0) |- ([yshift=-0.5cm]mmio.south) -| (interrupt.south)
        node[pos=-0.3, below, label style] {(not MMIO)};

    \draw[arrow] (mmio.east) -- (interrupt.west)
        node[pos=0.5, below, label style] {\textcolor{red}{recvResp}};

    \draw[arrow] (interrupt.north) -- ++(0,0) |- ([yshift=0.5cm]leak.north) -| (leak.north)
          node[pos=-0.3, above, label style] {\textcolor{blue}{takeInterrupt?}};;

\end{tikzpicture}
\caption{The $\mathit{Cycle{-}ISA}$ state machine. The driver directs the machine to take steps and decides when an interrupt is taken. The machine outputs declassified leakage and drives wires on the MMIO interface.}\label{fig:spec_state_machine}
\end{figure}

We extend our three-state $\mathit{Cycle{-}ISA}$ machine to support interrupts with a fourth \texttt{StepInterrupt} state, summarized in \Cref{fig:spec_state_machine}.
$\mathit{DriverOut}$ is extended to provide (to \texttt{StepInterrupt}) a Boolean argument for whether to trap: if true, an interrupt is pending, and interrupts are enabled, then the machine traps to the handler;  otherwise, no interrupt is taken. This staging lets the implementation choose at which instruction boundary to trap an interrupt.

\paragraph{Generality.}
\Cref{sec:eval:mem} discusses how this transformation generalizes to settings with multiple ``next states,'' such as verifying a processor independently of memory, where instruction- and data-memory requests are part of the external interface.
To allow speculative loads with secret-independent timing and addresses, the machine can be extended with, e.g., a \texttt{SendLoadReq addr} directive issuable by the secret-independent driver at any step.

\subsection{Determinizing Nondeterminism via Existential Parameters}\label{sec:existential_parameters}
Recall from \Cref{sec:overview:nondeterminism} that the cycle-level specification is parameterized by three untrusted black boxes; we give their types and their roles in the context of an ISA specification.

\paragraph{The driver} $\mathit{Driver} : \mathit{Machine}\ (\mathit{PubIn} \times \mathit{LeakEvent})\ \mathit{DriverOut}$, fed only public data, emits a sequence of commands to the ISA machine. 
In the baseline of \Cref{fig:spec_state_machine} a command is a list of Booleans dictating how many steps to take and whether to trap (ignored unless at \texttt{StepInterrupt}); when a state admits multiple transitions (e.g.\ exposing memory requests, c.f.\ \Cref{sec:eval:mem}) it selects the transition, such as whether to issue a speculative load.

\paragraph{The witness} $\mathit{Witness} : \mathit{Machine}\ (\mathit{PubIn} \times \mathit{SecIn})\ \mathit{WitnessOut}$ models unspecified, possibly secret-dependent behaviour---in our case study, the MMIO output wire on cycles when \texttt{valid} is low, which the consumer must ignore. \Cref{sec:modular} shows composition with a parent module removes witness-induced leakage from the TCB.

\paragraph{The leakage transformer} $\mathit{LT} : \mathit{Machine}\ (\mathit{PubIn} \times \mathit{LeakEvent})\ \mathit{PubOut}$ maps cycle-level leakage to adversary observations using public data only. 
Its existence is the security property: observations---including timing signals such as ready/valid---reveal nothing beyond what the contract declassifies.

\subsection{Functional Correctness and Nonleakage via Trace Equivalence}\label{sec:spec:refinement}
The above construction allows a single trace-equivalence theorem to discharge both properties at once (and a trivial transitivity property, useful for modular proofs).
Composing the pieces as in \Cref{fig:spec_overview}, the specification machine
$$\mathit{Spec}_{\{\mathit{LT,Driver,Witness}\}} : \mathit{Machine}\ (\mathit{PubIn} \times \mathit{SecIn})\ (\mathit{PubOut} \times \mathit{PubOut} \times \mathit{SecOut})$$ 
exposes two public-output channels: one driven by the leakage transformer (observations) and one by the ISA's functional-correctness definition.
Lifting the implementation $\mathit{\mu} : \mathit{Machine}\ (\mathit{PubIn} \times \mathit{SecIn})\ (\mathit{PubOut} \times \mathit{SecOut})$ to the same interface by duplicating its public-output wires lets both properties be discharged by a standard trace-equivalence theorem (this duplication is not essential; the equality can be asserted separately).

\begin{defn}[Correctness and nonleakage under a leakage-aware contract]\label{defi:correct}
    An implementation is functionally correct and nonleaking under a leakage-aware ISA if there exist leakage-transformer, driver, and witness machines such that the implementation and specification traces are equivalent (for all cycle counts $n$ and public and secret inputs, $\vec{\mathit{Pub}_n}$ and $\vec{\mathit{Sec}_n}$, of length $n$):
 \begin{align*}
    & \exists \mathit{LT}, \mathit{Driver}, \mathit{Witness}.\ \forall n, \vec{\mathit{Pub}_n}, \vec{\mathit{Sec}_n}. \\
    & \quad O_{\mu}^n(\vec{\mathit{Pub}_n},\vec{\mathit{Sec}_n})
        = \mathit{LT}^n(\vec{\mathit{Pub}_n}, O_{\mathit{Cycle{-}ISA}\{\mathit{Driver}, \mathit{Witness}\}}^n(\vec{\mathit{Pub}_n}, \vec{\mathit{Sec}_n}))
        \;\wedge \\
    & \quad \mathit{Func}_{\mu}^n(\vec{\mathit{Pub}_n},\vec{\mathit{Sec}_n})
        = \mathit{Func}_{\mathit{Cycle{-}ISA}\{\mathit{Driver}, \mathit{Witness}\}}^n(\vec{\mathit{Pub}_n}, \vec{\mathit{Sec}_n}),
\end{align*}
where $O_{\mu}$ denotes the adversary observation trace (of type $\mathit{list\  PubOutput}$), $\mathit{Func}$ the functional (I/O) trace (of type $\mathit{list\ (PubOutput \times SecOutput)}$), and $O_{\mathit{Cycle{-}ISA}}$ the specification leakage trace. Equivalently, as a single-trace equation:
 \begin{align*}
    & \exists \mathit{LT}, \mathit{Driver}, \mathit{Witness}.\ \forall n, \vec{\mathit{Pub}_n}, \vec{\mathit{Sec}_n}. \\
    & \mathit{Tr}^n_{\mu}(\vec{\mathit{Pub}_n}, \vec{\mathit{Sec}_n}) = 
    \mathit{Tr}^n_{\mathit{Spec}_{\{\mathit{LT,Driver,Witness}\}}} (\vec{\mathit{Pub}_n}, \vec{\mathit{Sec}_n}) 
\end{align*}
\end{defn}

\subsection{Classes of Bugs Ruled Out by the Top-Level Specification}
Proving trace equivalence between the top-level specification and synthesizable RTL rules out entire classes of microarchitectural and security bugs without enumerating individual attacks, yielding robustness to attacks-not-yet-discovered.

\paragraph{Functional-correctness bugs ruled out.}
Equivalence to a sequential, OIAAT ISA specification eliminates standard and edge-case correctness bugs:
incorrect instruction-decoding and execution-logic errors, data hazards and pipeline races, and imprecise traps (for precise exceptions and asynchronous interrupts, bugs that trigger mid-instruction or corrupt saved contexts; the \texttt{StepInterrupt} stage guarantees interrupts land only at instruction boundaries).

\subsubsection{Timing-side-channel and security bugs ruled out}
Timing leakage is an emergent property of subtle microarchitectural interactions that classical
functional testing ignores. \granite{}'s noninterference proof rules out vulnerabilities---without
naming speculated instructions or microarchitectural structures---such as:
\begin{itemize}
    \item Transient/speculative-execution attacks (e.g.~Spectre, Meltdown): side effects of transiently executed instructions (on misspeculated, unreachable paths) cannot leak secrets.
    Because the specification declassifies leakage only for instructions on correct paths, speculative state changes remain secret-independent; a secret-dependent speculative load or branch with observable timing effects would make it impossible to construct a secret-independent shadow machine with matching timing for all public data.
    \item Speculative interference and port-contention leaks: younger, misspeculated instructions cannot alter timing of older, bound-to-retire instructions via contention on ALUs, MSHRs, or functional units, since such instructions do not appear in the specification
    leakage trace.
    \item Secret-dependent interrupt timing leaks: classical refinement permits handling an interrupt immediately when a secret is zero but delaying it otherwise;
    \granite{} rules out this behaviour.
\end{itemize}

\subsection{Proof Strategy and Instantiating Existential Parameters}\label{sec:overview_proof_strategy}
\Cref{defi:correct} yields a formulaic proof strategy: construct a deterministic specification machine and cycle-level bisimulation relation, reducing the proof obligation to single-cycle preservation of a simulation relation between three machines---the $\mathit{Cycle{-}ISA}$ machine, the implementation, and the ``shadow machine'' (a copy of the implementation that runs only on public data).
We summarize the elements below (see \Cref{sec:case_study} for detail).

\paragraph{1) Instantiating the existential parameters}
The witness is simply a copy of the implementation machine.
The leakage transformer and driver can be instantiated with shadow copies of the implementation containing only public data.
As the shadow copies execute, they replace data dependent on ISA-declassified information with information from leakage events (e.g.\ the address of a load/store request is replaced with the address from the leakage trace).
The shadow machines differ from the implementation on, e.g., register-file contents and parts of microarchitectural state, but they will likely agree on aspects of microarchitectural state influencing control flow (such as branch-predictor state).
The specification does not prescribe what the shadow machine consists of, and there are multiple valid instantiations (for example, leakage information can be threaded through at any point before it affects observable timing).
The driver essentially precomputes whether, in the upcoming cycle, the next instruction reaches its visibility point (allowing a leakage event to be generated), an observable event occurs (e.g., sending/receiving an MMIO request/response), or it is resolved whether an interrupt is taken after the previously committed instruction.

\paragraph{2) Relating the implementation state to the ISA state machine.}
This element is standard in functional-correctness proofs in the style of Fjfj and Kami.
We state a simulation relation between specification and implementation state as a flushing invariant~\cite{su1996automatic}: an instruction bound-to-commit in the implementation relates to the specification's state after that instruction has committed.
Nondeterminism not-yet-resolved, such as whether an interrupt is taken, is delayed in the specification until the implementation resolves it.
In other words, the driver cannot drive choices with observable side effects (MMIO loads/stores, taking interrupts) until the implementation has resolved the choices (in this case, one would fail to construct a valid proof).

\paragraph{3) Relating the implementation state to the shadow machine.}
Given the shadow machine above, the implementation and shadow machine share the same microarchitectural state.
The relation amounts to designating which signals are public (typically signals related to control logic or that affect timing) and equating them---e.g., branch-predictor state, state derived from fetched instructions, and the valid bits of pipeline FIFOs. 
\Cref{sec:modular} showcases how an abstract, method-call-history representation of submodules supports a lightweight strategy of asserting that public state is equivalent across submodules based on asserting equivalence of public traces, independent of their concrete implementation.

\section{Horizontal Modularity via Modular Abstractions}\label{sec:modular}
\granite{} is also \emph{horizontally} modular: a design at the abstract hardware layer is decomposed into submodules---FIFOs, multipliers, branch predictors, the memory subsystem---that interact only through method calls.
Each submodule is specified by its own contract capturing both functional correctness and leakage, independent of the rest of the design.

\paragraph{One method at a time}
The key enabler is a \emph{one-method-at-a-time} abstraction, the method-call analogue of the one-rule-at-a-time discipline of rule-based HDLs such as Bluespec~\cite{bluespec-memocode04}.
Rather than reasoning about concurrently updating wires, \granite{} restricts intercomponent communication to method calls with a deterministic, sequential semantics: each call is evaluated atomically, and cycle-accurate timing is recovered by a \texttt{tick} method whose top-level invocation constitutes one clock cycle.
This structure lets each submodule be verified in isolation by simulation against its specification and then substituted into the larger design, so the whole-processor proof reasons over canonical representations and abstracts over implementation details of subcomponents.
This intercomponent modularity is formalized in Rocq using freer monads to represent component interactions, with the \emph{one-method-at-a-time} sequential abstraction. 

\paragraph{Substitution principle}
A \emph{substitution principle}, proven sound based on the one-method-at-a-time reasoning, allows substitution in proofs.

\begin{defn}[Leakage refinement]
Let $\mathit{SubImpl}$ be a Mealy machine and $\mathit{SubSpec}[\cdot]$ be a parameterized family of Mealy machines.
We say a $\mathit{SubImpl}$ \emph{leakage-refines} its local specification $\mathit{SubSpec}$ if there exist parameters $\mathit{SubParams}$ such that $\mathit{SubImpl}$ is equivalent to $\mathit{SubSpec[SubParams]}$ for all inputs, denoted as $\mathit{SubImpl} \sqsubseteq_{\exists} \mathit{SubSpec[\cdot]} \triangleq \exists \mathit{SubParams}.\ \mathit{SubImpl} \equiv \mathit{SubSpec[SubParams]}$.
\end{defn}

\begin{theorem}[Substitution principle]\label{thm:substitution}
Let $\mathit{SubImpl} \sqsubseteq_{\exists} \mathit{SubSpec[\cdot]}$ as above.
If a top-level implementation $\mathit{Impl}(\cdot)$ verified with submodule $\mathit{SubSpec}[\cdot]$ is equivalent to $\mathit{Spec}$ for all choices of $\mathit{SubParams}$, then the $\mathit{Impl}$ instantiated with the concrete submodule $\mathit{SubImpl}$ is equivalent to the specification:
\begin{align*}
    & \mathit{SubImpl} \sqsubseteq_{\exists} \mathit{SubSpec[\cdot]} \wedge \left( \forall \mathit{SubParams}.\, \mathit{Impl(SubSpec[SubParams])} \equiv \mathit{Spec} \right) \\
    & \rightarrow  \mathit{Impl(SubImpl)} \equiv \mathit{Spec} 
\end{align*}
Composing with parameterization of $\mathit{Spec}$ follows readily.
\end{theorem}

\subsection{Hardware Semantics}\label{sec:hardware_semantics}

To formalize modular hardware components within Rocq, \granite{} uses a lightweight program embedding structured as a freer monad, summarized in \Cref{fig:monads}.

\begin{figure}[h]
\begin{multicols}{2}
\begin{lstlisting}
Inductive prog {Method: Type -> Type} {Ret: Type} :=
| Return : Ret -> prog 
| Call {A} : Method A -> (A -> prog) -> prog.
Record Spec {Method: Type -> Type} :=
{ State : Type;
  EvalMethod {A} : Method A -> State -> A * State; 
  initialState: State }
Record Module {Method: Type -> Type} :=
{ ModuleMethod: Type -> Type;
  ModuleBase : Spec ModuleMethod;
  ModuleProgram: forall Ret, Method Ret -> 
                 prog ModuleMethod Ret }
Fixpoint evalProg 
  {Method} (spec: Spec Method) {R}
  (p: prog Method R) (st: spec.(State))
  : R * spec.(State) :=
  match p with
   | Return e => (e, st)
   | Call m f => let '(r, st') := spec.(EvalMethod) m st in
                 evalProg (f r) st' 
  end.
Definition refines Rel :=
  forall s1 s2, Rel s1 s2 ->
  (forall method r s1', spec1.EvalMethod method s1 = (r, s1') ->
            exists s2', spec2.EvalMethod method s2 = (r, s2') /\ 
                        Rel s1' s2').
Definition simulates (spec1 spec2: Spec Method) :=
  exists Rel, Rel spec1.initialState spec2.initialState $\wedge$ 
    refines Rel.
\end{lstlisting}
\end{multicols}
\caption{The freer-monad encoding of modular hardware in Rocq. For verification convenience (elided here), our Rocq implementation distinguishes value methods, which do not alter state, from action methods.}\label{fig:monads}
\end{figure}

A hardware program of type $\texttt{prog}$ is parameterized by a method signature $\texttt{Method}$ and return type \texttt{Ret} (lines 1-3).
A component specification is a record type \texttt{Spec} containing the internal state type, a per-method evaluation function, and an initial state (lines 4-7).
A hardware module is an instance of \texttt{Module}, associating an interface signature with a base specification and an implementation of each method (lines 8-12).
Programs are interpreted sequentially by invoking the base specification at each method call (lines 13-21) and can be lifted to define component specifications (lines 22-24).

Refinement between two component specifications sharing a method interface is a simulation relation over their state spaces: \texttt{Rel} is a simulation if, for any two related states, executing the same method yields equal output and related successor states (lines 25-29).
This definition gives rise to a substitution principle, where a parent module driving the same method sequence into the implementation and specification cannot tell them apart.

The encoding is independent of the underlying HDL, enabling language-agnostic reasoning and early-stage exploration. 
When state and programs are written in a synthesizable fragment of Gallina, they can be proven to correspond to designs in a deeply embedded HDL and synthesized to RTL, for an end-to-end theorem.

\subsection{Intercomponent Modularity}\label{sec:intercomponent_modularity}
Modules are specified as abstract specifications exposing method calls, formalized as Mealy machines whose inputs are method calls.
We adopt two modeling principles.
First, we use \emph{canonical} functional representations (e.g., a hardware FIFO realized as a circular buffer has many states corresponding to a FIFO with one element; we represent FIFOs as lists).
Second, we extend the state representation with a \emph{leakage history} that projects the trace of method calls (inputs) down to the information permitted to influence timing (the public inputs)\footnote{The leakage history could also carry data declassified during execution, as in the top-level ISA spec, but our case studies did not require such declassification.}.
Submodules are existentially parameterized with the same parameters in \Cref{sec:existential_parameters}: 
 a secret-independent driver over the module's public leakage history fixes its timing, and an unconstrained witness supplies its unspecified outputs.

\subsubsection{Example: zero-skip multiplier}\label{sec:multiplier}
Recall the multiplier from \Cref{sec:overview:example}.
\Cref{listing:multiplier} gives the Rocq spec.

\begin{figure}[h]
\begin{multicols}{2}
\begin{lstlisting}
Record req_t := { a : bv width; b : bv width }.
Inductive Method: Type -> Type :=
| Enq (arg: req_t) : Method unit
| Deq              : Method unit
| Tick             : Method unit
| RespReady        : Method bool 
| Full             : Method bool 
| Peek             : Method (bv (width + width)). 
Inductive LeakEvent := 
| LeakEnq (zeroArg: bool) | LeakDeq | LeakTick.
Definition leakTrace_t := list LeakEvent.
Definition trace_t := list (Method unit).
Class specParams :=
{ default_peek: trace_t -> bv (width + width)
; resp_ready : leakTrace_t -> bool 
; is_full : leakTrace_t -> bool }.
Definition leakage (tr: trace_t) : leakTrace_t :=
  map (fun m => 
    match m with 
    | Enq req => LeakEnq (req.a = 0 || req.b = 0) 
    | Deq => LeakDeq | Tick => LeakTick end) tr.
Record St := { reqs: list req_t; hist : trace_t }.
Definition enq (req: req_t) (st: St) :=
  {| reqs := if is_full (leakage st.hist) then st.reqs 
             else st.reqs $++$ [req];
     hist := st.hist $++$ [Enq req] |}.
Definition deq (st: St) :=
  {| reqs := if resp_ready (leakage st.hist)
             then list.tail st.(reqs) else st.(reqs);
     hist := st.hist $++$ [Deq] |}.
Definition tick (st: St) :=
  {| reqs := st.reqs;
     hist := st.hist $++$ [Tick] |}.
Definition full (st: St) := params.is_full (leakage st.hist).
Definition respReady (st: St) :=
    match st.(reqs) with   
    | [] => false (* for functional correctness *)
    | _ => resp_ready (leakage st.hist) end.
Definition peek (st: St) :=
  if respReady (leakage st.hist) then
    match st.reqs with 
    | [] => default_peek st.hist 
    | req::_ => req.a * req.b end 
  else default_peek st.hist.
Definition evalMethod {A} (m: Method A) (st: St) : A * St :=
  match m with
  | Enq req => ((), enq req st)
  ...
  | Full => (full st, st) end.
\end{lstlisting}
\end{multicols}
\caption{Zero-skip multiplier specification in Rocq. NB: simplified to remove value/action method distinction.}\label{listing:multiplier}
\end{figure}

\paragraph{Functional correctness.}
The multiplier is functionally correct as long as it computes the product of its inputs and returns results in FIFO order.
\texttt{Full} and \texttt{RespReady} act as the ready/valid handshake.
Unlike conventional Bluespec models, where method conflicts block illegal transitions, here \texttt{Enq} and \texttt{Deq} may always fire---enabling purely local reasoning about each method, which simplifies verification; we synthesize circuits from this style in \Cref{sec:integration}.
It is therefore the  caller's responsibility to check \texttt{Full} before an \texttt{Enq}, and \texttt{RespReady} before a \texttt{Deq} or using the result of a \texttt{Peek}.

\paragraph{Security: the leakage contract.}
A multiplier implementation is secure if its timing behaviour is a deterministic function of the specification's declassified leakage trace---here, the history of action-method calls (\texttt{Enq}, \texttt{Deq}, and \texttt{Tick}) with each \texttt{Enq} recorded along with whether either input argument is zero.
Zero-skip optimization is permitted because that bit is logged in the trace; all other latency variation must be a function of the trace, ruling out operand-dependent timing channels.
It is the caller's responsibility to ensure that the zero-operand bit is indeed public.

\paragraph{Specification structure.}
The specification state tracks a FIFO of outstanding requests and a history of past operations.
To admit a range of implementation choices (varying pipeline depths, multiplication algorithms), it is parameterized over two driver functions of only the public leakage trace, \texttt{resp\_ready} and \texttt{is\_full}, so response time depends exclusively on public or explicitly declassified data.
To model unspecified behaviour, such as the value of \texttt{Peek} when \texttt{RespReady} is false (corresponding to the data wire while \texttt{valid} is low), the specification includes a witness parameter \texttt{default\_peek}, letting that output vary or even hold secret-dependent values while invalid and shifting the obligation to the consumer to ignore it until \texttt{valid}.
The \texttt{tick} method marks clock-cycle progression, giving implementations a point to update internal state.

\paragraph{Proof structure}
A concrete implementation realizes every method using a synthesizable fragment of Gallina (e.g.~using bitvectors and vectors rather than naturals and lists). 
We provide both a shallow-embedded and a deeply-embedded-HDL implementation (\Cref{sec:case_study}).
As in \Cref{sec:overview_proof_strategy}, each existential parameter is instantiated with a shadow copy of the implementation: the witness \texttt{default\_peek} runs the implementation on the method-call trace, while the secret-independent parameters \texttt{resp\_ready} and \texttt{is\_full} run shadow machines on \emph{lifted} method calls:
\begin{lstlisting}
Definition lift_leakage_event (ev: LeakEvent) : Method unit :=
  match ev with
  | LeakDeq => Deq | LeakTick => Tick 
  | LeakEnq zero_arg =>
      Enq (if zero_arg then  {| a := zeroes; b := zeroes |} else  {| a := ones; b := ones |} )                     
  end.
\end{lstlisting}
The proof writer then states and proves a simulation relation between the implementation and both the specification and shadow machines.
Against the specification, the relation maps valid implementation states to the canonical list-of-requests representation.
Against the shadow machine, it relates public information: equality of public registers and ``leakage equivalence'' of states (e.g.~data derived from requests need not be equal, only equal modulo whether either operand was zero).
The simulation lets the top-level proof treat the multiplier as a black box, substituting any compliant implementation, and allowing a high-level simulation relation in terms of the abstract multiplier's state (FIFOs and traces of method calls).

\subsubsection{Specifying the memory}\label{sec:memory}
Memory is modeled abstractly in a similar style and with the same API as the multiplier (with a different enqueue request type and peek response type).
Functionally, it is a map from address to byte together with a FIFO of requests and the history of method calls updating state (enqueue, dequeue, and tick).
This pattern is standard for pipelined submodules that service requests in order, one at a time\footnote{We can extend to out-of-order modules by relaxing the FIFO nature of the request queue and adding a secret-independent parameter to choose which request is handled next.}.
As with the multiplier, response time is asserted to be a function of public information; here the leakage trace includes addresses of requests but not data. 
This work does not tackle verifying memory hierarchies, and so, the spec of memory is trusted.

\subsection{Intracomponent Modularity}
Within a single component such as the processor core, the \texttt{tick} cycle-update function is decomposed into sequential steps for individual pipeline stages (fetch, decode, execute, writeback). 
This lets designers and verification engineers establish Hoare-style pre-/postconditions at stage and function boundaries---structuring each stage to preserve a simulation relation, and supporting both high-level and optimized low-level implementations of functions.

\section{Verification Case Study: A Pipelined Processor}\label{sec:case_study}
We verify a four-stage pipelined core implementing a representative subset of RISC-V---with arithmetic, memory, control, and CSR instructions---against the cycle-accurate specification of \Cref{sec:spec:hwsw}, composed with the model of memory of \Cref{sec:memory}.

\paragraph{Contract bug found in our work.}
We initially verified our processor directly against the constant-time policy: it declassifies load/store addresses, branch conditions, and whether a multiplier operand is zero. 
Verifying hardware against the contract exposed a declassification our software-level intuition had missed: writes to the interrupt-configuration CSRs (\texttt{mie}, \texttt{mtvec}) must themselves be leaked, because the processor jumps to \texttt{mtvec} on an interrupt, and thus the trap target influences observable control flow. 
Prior work verifying constant-time software down to RISC-V~\cite{bedrock-ct} did not account for this leakage, because reasoning at the ISA level alone does not force the question.
The proof could not be completed without the extra leakage clause, demonstrating that verifying hardware against ISA specs catches subtle, easy-to-miss policy omissions.

\paragraph{Processor stages}
The processor is a standard fetch-decode-execute-writeback pipeline: fetch speculates instruction loads using a BTB-predicted PC and an epoch tag; decode checks the epoch, stalls on scoreboard hazards (and until older instructions commit before a CSR access) then reads registers; execute resolves branches (bumping the epoch and flushing on mispredict) and dispatches variable-latency work to the multiplier and memory; writeback commits results and takes precise exceptions and interrupts (flushing the pipeline if need be).
The case study's focus is not on the pipeline but on aspects connecting it to the spec: the visibility point, shadow core, and driver.

\paragraph{Identifying the visibility point}
The specification is one-instruction-at-a-time, but the implementation has many instructions in-flight.
The \emph{visibility point} of an instruction is the cycle at which it becomes bound-to-commit---no longer discardable by an older instruction's mispredict, exception, or interrupt.
It tells the driver when to advance the specification (for functional-correctness proof) and it marks when an instruction's leakage (e.g.\ a branch direction) is safe to declassify to the shadow core.
It is sufficient to define a visibility point (which could differ based on instruction type) such that any leakage information that could influence execution time in the shadow machine does so strictly after the visibility point.
Then, the leakage information is available to be used early enough by the shadow core to replicate the implementation's timing behaviours.

In our core, we safely define the visibility point as a simple predicate on occupancy of pipeline queues: an instruction is visible when it is at the head of the decode-to-execute (\texttt{d2e}) queue, there are no instructions at the execute-to-writeback (\texttt{e2w}) queue, and it is at the correct epochs.
With \texttt{e2w} empty, no older instruction remains that could trap or be interrupted. 
This visibility point suffices as leakage information is not needed in the shadow machine until the execute stage.
Tighter, state-dependent conditions are possible (e.g.\ a lone, in-flight instruction with correct epochs is always correct-path) but not necessary.

\paragraph{Driving nondeterministic choices}
Analogously, the driver should not drive the specification past nondeterministic choices influencing observables (MMIO, interrupts) until the implementation resolves nondeterministic choices.
For example, the driver should not drive the spec to not take an interrupt until the implementation has resolved to not take an interrupt after the corresponding instruction in the spec. 
The driver is untrusted: a poor choice of driver will not lead to bugs in the security guarantee but simply a failure to prove the theorem.

\paragraph{Constructing the shadow core.}
To prove that the processor's timing is independent, \granite{}'s framework uses a public shadow core that replicates the processor's control logic but operates only on public data.
Data memory is initialized with constant dummy values (e.g.~zero), with data affecting control flow driven by the public leakage trace.
Before the execute stage, the shadow core runs identically to the real core: instructions are fetched speculatively based on the public $\pc$ and public branch-predictor state, and registers read from the shadow core in the decode stage contain dummy data.
At execute, it substitutes leakage from the secret-dependent decisions:
\begin{enumerate}
    \item The output of control logic determining whether a branch is taken for a control instruction is replaced with the Boolean from the leakage trace.
    \item The address of a memory load/store is replaced with the address from the leakage trace.
    \item Arguments to the multiplier are reconstructed based on a Boolean of whether either argument is zero (e.g.\ by passing in zeroes if true and ones if false).
    \item Writes to the \texttt{mie} and \texttt{mtvec} CSRs are replaced accordingly in the shadow machine.
\end{enumerate}
Because every decision that eventually affects timing is replaced by public leakage, the shadow core's timing is a function of only public data.

\paragraph{Relating the implementation and ISA machine}
The driver keeps the specification in-sync with the implementation, advancing it at visibility points as follows:
\begin{itemize}
    \item When an instruction leaves \texttt{e2w} to commit next cycle (its multiplier/memory responses are ready), the specification is at \texttt{StepWaitMMIOResp} or \texttt{StepInterrupt}. The driver advances the specification two steps for an MMIO response and one otherwise, passing the interrupt bit and transitioning to \texttt{StepLeak}.
    \item When an instruction advances from \texttt{d2e} to \texttt{e2w} (hence correct-path), all older instructions have committed, and the specification is at \texttt{StepLeak}. The driver advances the spec two steps (corresponding to \texttt{StepLeak} and \texttt{StepInstr}), ensuring leakage information is available before the implementation executes the instruction.
\end{itemize}

The implementation and ISA state machine are related by a flushing relation~\cite{su1996automatic}: an instruction in \texttt{e2w} is related to the specification after it commits (the spec machine already executed the instruction and updated the register files), and unresolved nondeterminism (e.g.\ whether an interrupt is taken) is delayed in the specification until the implementation resolves it. 

\paragraph{Relating the implementation and shadow machine}
The simulation relation between the implementation and shadow machine equates state with downstream effects on timing. 
For example:
\begin{itemize}
    \item Memory/multiplier: request histories are equal modulo store data (no effect on timing) and zero-operand quotient, respectively, forcing identical latencies.
    \item Pipeline FIFOs: equal occupancy and valid bits (hence identical stalls) and equal public bookkeeping state (such as instructions and epoch bits)
    \item Branch predictor: updated only from public information, so its state (modeled as a trace of method calls) stays public. This invariant is used to guarantee that the PC and speculated instructions remain public.
\end{itemize}

Notably, specifying leakage components of modules in terms of traces of method calls enables a streamlined recipe for stating simulation relations that is independent of underlying implementations.
We prove simulation relations are preserved at each fetch-decode-execute-writeback stage boundary of the processor, for all submodule implementations satisfying their respective specifications (recall \Cref{thm:substitution}).

\section{Integration Verification}\label{sec:integration}
This section describes how software can be verified against the ISA specification via a certified static analysis and extends the processor proof down to RTL via \quartz{}, yielding an end-to-end proof that eliminates the ISA specification from the TCB. 
The static analysis and HDL are not the paper's focus but serve to validate the ISA spec and our methodology for integration verification, culminating in a noninterference theorem for a constant-time Salsa20 binary down to RTL.

\subsection{Software Static Analysis}\label{sec:static_analysis}
We implement a sound (not complete) ISA-level static analysis deciding whether a RISC-V binary is constant-time. 
Implementation and proof of the analysis took $\sim$2 days.

\begin{defn}[Static analysis sound]
$ \forall \mathit{prog}. \; \mathit{analyze}(\mathit{prog}) \Downarrow \mathit{Safe} \rightarrow \mathit{IsConstantTime}_{\mathit{ISA}}(\mathit{prog}),$
where $\Downarrow \mathit{Safe}$ means the analysis terminates with outcome $\mathit{Safe}$.\footnote{In Rocq, formalized as there existing some fuel}
\end{defn}

\begin{defn}[Constant time]
A program is constant-time if for all instantiations of the existential parameters, the ISA-level leakage trace is secret-independent.
\begin{align*}
\mathit{IsConstantTime}_{\mathit{ISA}}(prog) \triangleq  & \forall \mathit{Driver}, \mathit{Witness}, n, \vec{\mathit{Pub}_n}, \vec{\mathit{Sec^n_1}}, \vec{\mathit{Sec^n_2}}. \\
&O_{\mathit{Cycle{-}ISA}\{\mathit{Driver}, \mathit{Witness}\}}^n(\vec{\mathit{Pub}_n}, \vec{\mathit{Sec}_n^1}) =
O_{\mathit{Cycle{-}ISA}\{\mathit{Driver}, \mathit{Witness}\}}^n(\vec{\mathit{Pub}_n}, \vec{\mathit{Sec}_n^2}) 
\end{align*}
\end{defn}

\begin{theorem}[End-to-end noninterference]\label{thm:e2e}
Let $\mu$ be an implementation satisfying \Cref{defi:correct}. For any program with
$\mathit{analyze}(\mathit{prog})\Downarrow\mathit{Safe}$,
$$\forall n, \vec{\mathit{Sec}_n^1}, \vec{\mathit{Sec}_n^2}.\;
  O_{\mu}^n(\vec{\mathit{Pub}_n},\vec{\mathit{Sec}_n^1}) = O_{\mu}^n(\vec{\mathit{Pub}_n},\vec{\mathit{Sec}_n^2}).$$
\end{theorem}
\begin{proof}[Proof sketch]
Soundness of the static analysis gives $\mathit{IsConstantTime}_{\mathit{ISA}}(\mathit{prog})$, so the ISA leakage trace is secret-independent;
since the leakage transformer reads only public data and that trace, its output $\mathit{LT}^n(\vec{\mathit{Pub}_n}, O_{\mathit{Cycle{-}ISA}})$ is also secret-independent; and \Cref{defi:correct} equates this output with $O_\mu$. 
Noninterference is a trace property preserved by trace equivalence, so the result is preserved under successive trace-equivalence proofs down to RTL (\Cref{sec:integration:rtl}).
\end{proof}

\paragraph{Static analysis.}
The static analysis symbolically executes the program over states that tag each register, CSR, and byte in memory as $\mathtt{Public\,n}$ or $\mathtt{Secret}$ (imprecise but sufficient in our case study, where the only initial public data is instruction memory), ensuring all branch conditions, arguments to variable-latency instructions, and addresses of memory accesses are independent of \texttt{Secret} values.
Our analysis is not designed to handle infinite loops (it terminates soundly on the spin-loop pattern \texttt{Beq x0 x0 0}) or interrupt/exception handlers.
\Cref{fig:analysis} shows the state representation and the representative \texttt{Beq} rules; a secret operand is rejected (\textsc{Beq-Secret}), a public branch advances the symbolic PC (\textsc{Beq-Taken}), and a misaligned branch target is rejected (\textsc{Beq-Exn}).
The remaining rules are analogous.

\begin{figure}
{\footnotesize
\begin{mathpar}
\texttt{SymbVal } \tau ::=  \texttt{Public (n: } \tau \texttt{) | Secret} \and
\texttt{RegFile = Reg} \rightarrow \texttt{SymbVal Word} \and 
\texttt{Memory = Addr} \rightarrow \texttt{SymbVal Byte} \\
\texttt{st} \in \mathrm{SymbState} = \{ \texttt{Pc : Addr,  Rf : RegFile,  Mem : Memory, Csrs : CsrFile} \} \\
\texttt{res} \in \mathrm{SymbResult} ::= \texttt{Safe | Unsafe | Running(st)} \\
\and
\inferrule*[right=Beq-Spin]
  {\mathtt{rs1 = 0} \\ \mathtt{rs2 = 0} \\ \mathtt{offset = 0} \\ \mathtt{interruptsDisabled(st)}} 
  {\langle \texttt{Beq rs1 rs2 offset, st} \rangle \rightarrow \mathtt{Safe}} 
\and
\inferrule*[right=Beq-Taken] 
 { \neg \texttt{Spin} \\   \texttt{st.Rf(rs1) = Public v1} \\  \texttt{st.Rf(rs2) = Public v2} \\ \texttt{v1 = v2} \\ \mathtt{aligned(st.Pc + offset)}}
 {\langle \texttt{Beq rs1 rs2 offset, st} \rangle \rightarrow \mathtt{Running (st[Pc \leftarrow st.Pc + offset])}} 
\and 
\inferrule*[right=Beq-Exn] 
 {\neg \texttt{Spin}  \\ \texttt{st.Rf(rs1) = Public v1} \\  \texttt{st.Rf(rs2) = Public v2} \\ \texttt{v1 = v2} \\ \neg \mathtt{aligned(st.Pc + offset)}}
 {\langle \texttt{Beq rs1 rs2 offset, st} \rangle \rightarrow \mathtt{Unsafe}}
\and 
\inferrule*[right=Beq-Not-Taken] 
 { \neg \texttt{Spin}  \\  \texttt{st.Rf(rs1) = Public v1} \\  \texttt{st.Rf(rs2) = Public v2} \\ \texttt{v1 != v2} }
 {\langle \texttt{Beq rs1 rs2 offset, st} \rangle \rightarrow \mathtt{Running (st[Pc \leftarrow st.Pc + 4])}} 
\and
\inferrule*[right=Beq-Secret] 
 {\neg \texttt{Spin} \\ \texttt{st.Rf(rs1) = Secret} \vee \texttt{st.Rf(rs2) = Secret}} 
 {\langle \texttt{Beq rs1 rs2 offset, st} \rangle \rightarrow \mathtt{Unsafe}} 
\end{mathpar}
}
\caption{Symbolic state and representative \texttt{Beq} rules of the constant-time analysis.}
\label{fig:analysis}
\end{figure}

\paragraph{Salsa20 example}
We compile an existing constant-time C implementation of Salsa20 to RISC-V with a standard toolchain, treating data memory (key, message, nonce) as secret and the program as public, and run the analysis in under a second. 
\Cref{thm:e2e} yields an end-to-end noninterference proof for the binary on our processor.
This proof rules out bugs in intermediate layers, including any bugs in our encoding of the RISC-V spec, and also does not require trusting the RISC-V compiler.

\subsection{Lowering to RTL via \quartz{}}\label{sec:integration:rtl}
Our top-level proofs use a shallow embedding in a synthesizable fragment of Gallina, describing a cycle-accurate design as a first-order functional program using operations commonly available in hardware-description languages.
To validate this design choice, we also give this fragment a deeply embedded syntax we call \quartz{}, and prove the design thus expressed equivalent to the Gallina-native version.
In \quartz{}, functions are written in terms of combinational expressions and let-expressions, methods are functions that take state and input as arguments and return state and output, modules are collections of methods operating on the same type.
This equivalence proof is straightforward, only deviating from a partial evaluation of the \quartz{} interpreter to reconcile encoding details such as use of native Gallina structs vs. deeply embedded structs interpreted as tuples (as Gallina does not support polymorphism over lists of struct fields).

The \quartz{} implementation of the processor is pretty-printed as SystemVerilog, generating reasonably readable code where variable names are preserved and e.g. source-level structs are represented as (packed, synthesizable) SystemVerilog structs.
While we consider it desirable to verify the translation, we are not aware of an adequately featureful formalization of SystemVerilog semantics against which it could be performed.
Thus the 200-line syntax-mapping function is currently part of our TCB.
We confirmed that the extracted top-level cycle function (and a trivial always-block wrapper) can be successfully synthesized for the ECP5 FPGA using open-source tools Yosys and NextPnR~\cite{YosysNextPNR},
but we have not integrated it into a SoC design for post-synthesis testing.

To prove another RTL description of the same cycle-update function in, e.g., Verilog is equivalent to our implementation, one could use automated tools for equivalence checking at RTL.
The proof obligation amounts to proving equivalence at a single cycle, which is less likely to run into challenges related to state-space explosion.

\section{Discussion}

This section discusses what is in the trusted computing base (\Cref{sec:eval:limits}); what changes to the spec and proof were needed when adding branch predictors, exceptions, and interrupts (\Cref{sec:eval:mods}); and how the methodology applies to verifying the processor independently of the memory subsystem (\Cref{sec:eval:mem}).

\subsection{Trusted Computing Base}\label{sec:eval:limits}
The TCB contains the $Cycle{-}ISA$ contract specifying functional and leakage semantics of a fragment of RISC-V ($\sim$700 LoC), the observation/adversary functions defining the I/O interface, the abstract model of the memory hierarchy, the model of the external world, the Rocq kernel, and the \quartz{} pretty-printer from circuits to Verilog.
The RTL processor, leakage transformer, driver, and every submodule implementation are not in the TCB; they are eliminated via machine-checked proof.
When connecting with a software proof, such as via a verified software static analysis, the $Cycle{-}ISA$ is eliminated from the TCB.
Intermediate specification layers are not in the TCB.

\subsection{Design Modifications}\label{sec:eval:mods}
Our initial design and proof had no branch predictors nor support for exceptions, interrupts, and CSRs.
We discuss how the proofs changed as these features were added.

\paragraph{Adding a branch predictor.}
Adding a branch history table and branch target buffer took < 1 day with no changes to the spec.
Predictors are specified as abstract state machines whose state is public (defined as a trace of method calls).
As our processor enforces that branch predictors are only updated with public information, the branch histories in the implementation and shadow machine are equivalent.
As there is no functional-correctness obligation, instantiating the branch-predictor specs with concrete implementations is trivial.

\paragraph{Adding exceptions and interrupts}
This extension required modifying the spec, adding CSRs for handling exceptions, exception semantics for instructions, and an interrupt-handling step.
Additionally, it required extending the driver commands from simply outputting the number of steps to including whether an interrupt is taken at the interrupt-handling step.
Notably, adding interrupts moved the visibility point later: an instruction in the execute-to-writeback queue may now be followed by an interrupt that flushes the pipeline, so the driver and invariants were updated to treat an instruction as visible only once no older instruction can cause a flush.

\paragraph{A multicycle processor}
To validate that the specification supports different processor architectures, we also verified a three-stage multicycle design, which required a significantly different top-level proof but no meaningful changes to the spec.

\subsection{Independent Verification of the Processor}\label{sec:eval:mem}
Fjfj~\cite{Fjfj} verifies a processor independently of the memory subsystem; the core question is specifying allowed load/store behaviour, since a processor may speculatively issue loads absent from the sequential ISA and reorder loads and stores. 
These speculative loads are permitted when the speculative addresses and the decision to issue are secret-independent.
We give an example proof of a processor specified independently of memory, allowed to have speculative loads:
the processor specification has load buffers and is existentially parameterized with a driver that can instruct it to take a normal step, issue a load request, or receive a response. 
As the driver is a function of only public information, implementation observations remain a function of public information---also demonstrating how driver commands generalize to specs with different next-state transitions.
This specification assumes no functional correctness of memory, specifying allowed behaviours of the processor in the face of arbitrary memory behaviour.
As in Fjfj~\cite{Fjfj}, one should be able to prove that composing with a well-behaved memory yields the top-level ISA spec, though the challenge is that the two specifications' leakage traces differ. 

We found that proving a processor while assuming an abstract model of memory allowed for a more satisfying specification without many additional assumptions.
An implementation processor would still be free to issue loads and reorder stores speculatively, as these are internal to the system and invisible at the observation level.
A specification of the processor independently of assumptions about memory has to specify what memory reorderings are allowed.
This processor spec may be important for multicore semantics and weak memory models, but it is unnecessarily complicated for single-core designs.
Nevertheless, we provide an example of how one could implement a specification machine with speculative loads and load buffering, while encapsulating nondeterminism and specifying leakage semantics.

\section{Related Work}

\paragraph{Verifying leakage contracts with model checkers.}
LeaVe~\cite{Leave}, UPEC and UPEC-DIT~\cite{UPEC,upec-dit}, and Contract Shadow Logic~\cite{contractShadowLogic} use model checkers to verify speculative constant-time contracts on Verilog.
However, they assume functional correctness (so the RTL stays in the TCB),
they do not natively model the nondeterminism of interrupts and I/O (typically assuming none occur), and their SMT-based property languages, effective for bounded relational checks, are less suited for high-level software specifications and compositional statements combining functional correctness and nondeterminism.
In contrast, \granite{} verifies both functional correctness and security of RTL designs against high-level hardware-software specifications encompassing nondeterminism, obtaining an easy-to-audit proof that composes with software proofs for an end-to-end guarantee.

\paragraph{Functional correctness via refinement and rule-based HDLs.}
Kami~\cite{kami} and Fjfj~\cite{Fjfj} prove functional correctness of processors written in rule-based HDLs.
However, their semantics abstracts away cycle-accurate timing, and they do not address leakage.
Lightbulb~\cite{lightbulb} extends Kami with integration verification from software to hardware but likewise does not target leakage.
MTIsolation~\cite{mtisolation} uses \Koika{}~\cite{koika-pldi20}, a rule-based HDL with cycle-accurate semantics, to prove enclave timing isolation via a hybrid Rocq--SMT approach and, like \granite{}, uses existential parameters for unspecified behaviour.
\granite{} extends these ideas with nondeterminism, processor functional correctness, and the dynamic declassification of HW/SW leakage contracts.

\paragraph{Information-preserving and secure refinement.}
Similar to our security goal, Knox~\cite{Knox} introduces information-preserving refinement (IPR) for hardware security modules but targets transaction-level calls rather than ISA leakage contracts; Parfait~\cite{parfait} proves IPR by collapsing hardware-software boundaries but for a specific program on a minimal core (PicoRV32).
Pantomime~\cite{pantomime} also derives the implementation leakage trace from a specification trace via a simulator, but its specification trace comes from a circuit-level description of complexity comparable to the implementation, whereas \granite{} relates RTL to an instruction-set-level leakage trace.
Correnson et al.~\cite{hoffmann2026deductive} give a deductive system for a property analogous to contractual noninterference but assume deterministic semantics and do not verify RTL.

\paragraph{Abstract models, leakage models, and mitigations.}
Pensieve~\cite{pensieve}, Guarnieri et al.~\cite{guarnieri2021hardware}, and CheckMate~\cite{checkmate} target abstract processor models.
These models are useful for early-stage designs and defenses but underspecify the cycle-level implementation details that can lead to leakages (Pensieve shares our decomposition of functional and leakage submodule specs but does not prove implementations satisfy them and is bounded to nine steps).
Motivated by transient-execution attacks, richer leakage models~\cite{ct-foundations,axiomatic-hw-sw-isca-22} with fuzzing-based validation~\cite{Revizor,hideAndSeek}, processors designed for secure speculation such as ProSpeCT~\cite{prospect}, defenses such as GhostMinion~\cite{ghostMinion}, and compilers such as Serberus~\cite{serberus} all aim to close the leakage gap but remain unverified at the RTL.

\paragraph{Noninterference and information flow.}
Various works on verifying the seL4 kernel~\cite{sel4}, mCertiKOS hypervisor~\cite{certikos-sec-pldi16}, and Komodo security monitor~\cite{komodo-sosp17} have formalized noninterference and information-flow security at the ISA level, but do not tackle RTL verification.
HyperFlow~\cite{hyperflow-ccs18}, SecVerilog~\cite{secverilog-asplos15}, and SpecVerilog~\cite{specverilog} enforce security with information-flow type systems: these approaches are automated but face completeness limitations common in static-analysis approaches.
For soundness, type systems are conservative and may reject secure designs because the type systems cannot reason about the functional logic that prevents leakage.
For expressivity, designers use declassification policies to express allowed leakages.
However, when defined at a low level associated with a particular implementation's code, understanding the security guarantee requires detailed auditing and deep understanding of the system.
We propose a spec that is auditable without understanding low-level details of the implementation.

\paragraph{Security verification under nondeterminism.}
O'Neill et al.~\cite{interactive-information-flow-chong} use \emph{refiners} to determinize nondeterministic semantics for an information-flow noninterference definition---analogous to our existential parameters, but their focus is on information-flow static analyses rather than specifications and proofs, and they do not address the challenge of declassification.
Conoly et al.~\cite{bedrock-ct} use \emph{predictors} to prove a compiler preserves constant time under allocation and I/O nondeterminism, also via leakage transformers, but do not tackle anything analogous to the underspecification of the mapping of instruction boundaries.

\section{Conclusion}
\granite{} shows that a particular combination of specification and proof techniques allows functional correctness and nonleakage guarantees to be achieved in a modular fashion, for the first time composing timing-security proof of digital-logic modules all the way up to the instruction-set-level specification (and then with software proofs).
More broadly, \granite{} suggests that the decades of refinement-based functional-correctness technology built for proof assistants can be repurposed, largely intact, for microarchitectural timing security---provided nondeterminism is determinized rather than abstracted away.

\bibliographystyle{ACM-Reference-Format}
\bibliography{refs}

\end{document}